\documentclass[journal]{IEEEtran}
\usepackage{cite}
\ifCLASSINFOpdf
  \usepackage[pdftex]{graphicx}
  \graphicspath{{pics/}{..}}
\else
\fi
\usepackage[cmex10]{amsmath}

\usepackage{amssymb,amsthm}

\ifCLASSOPTIONcompsoc
  \usepackage[caption=false,font=normalsize,labelfont=sf,textfont=sf]{subfig}
\else
  \usepackage[caption=false,font=footnotesize]{subfig}
\fi

\newtheorem{theorem}{Theorem}[section]
\newtheorem{lemma}{Lemma}[section]
\newtheorem{algorithm}{Algorithm}[section]

\newtheorem{remark}{Remark}[section]

\newtheorem{example}{Example}[section]

\newcommand{\T}[0]{\mathsf{T}}

\newcommand*{\diff}{\mathop{}\!\mathrm{d}}

\DeclareMathOperator{\Cov}{Cov}
\DeclareMathOperator{\Var}{Var}
\DeclareMathOperator{\E}{E}
\DeclareMathOperator{\N}{N}

\newcommand{\vmu}[0]{\boldsymbol{\mu}}

\newcommand{\vxi}[0]{\boldsymbol{\xi}}
\newcommand{\vtheta}[0]{\boldsymbol{\theta}}

\newcommand{\ve}{\mathbf{e}}
\newcommand{\vf}{\mathbf{f}}
\newcommand{\vg}{\mathbf{g}}
\newcommand{\vh}{\mathbf{h}}

\newcommand{\vk}{\mathbf{k}}

\newcommand{\vm}{\mathbf{m}}

\newcommand{\vo}{\mathbf{o}}

\newcommand{\vq}{\mathbf{q}}
\newcommand{\vr}{\mathbf{r}}

\newcommand{\vx}{\mathbf{x}}
\newcommand{\vy}{\mathbf{y}}

\newcommand{\MC}{\mathbf{C}}
\newcommand{\MD}{\mathbf{D}}

\newcommand{\MG}{\mathbf{G}}

\newcommand{\MI}{\mathbf{I}}

\newcommand{\MK}{\mathbf{K}}

\newcommand{\MP}{\mathbf{P}}
\newcommand{\MQ}{\mathbf{Q}}
\newcommand{\MR}{\mathbf{R}}
\newcommand{\MS}{\mathbf{S}}

\usepackage{color}

\begin{document}
\title{On the relation between Gaussian process quadratures and sigma-point methods}

\author{Simo S\"arkk\"a\thanks{The first author is grateful to the Academy of Finland for financial support.}\thanks{There are no conflict-of-interest or financial disclosure statements to be made at this time}\thanks{Authors' addresses: S. S\"arkk\"a, Aalto University, Rakentajanaukio 2 c, 02150 Espoo, Finland, (simo.sarkka@aalto.fi); J. Hartikainen, Rocsole Ltd., 70150 Kuopio, Finland, (jouni.hartikainen@rocsole.com); L. Svensson, Chalmers University of Technology, SE-412 96 Gothenburg, Sweden, (lennart.svensson@chalmers.se); F. Sandblom, Volvo Group Trucks Technology, SE-405 08 Gothenburg, Sweden, (fredrik.sandblom@volvo.com).},
Jouni Hartikainen,
Lennart Svensson, and
Fredrik Sandblom}%
%
%
%
%
%

%
%
%
%
%
%
%
%

%
%
%
%
%
%
%
%
%
%
%
%
%
%
%
%
%

%
\markboth{}%
{}
\maketitle

\begin{abstract}
This article is concerned with Gaussian process quadratures, which are numerical integration methods based on  Gaussian process regression methods, and sigma-point methods, which are used in advanced non-linear Kalman filtering and smoothing algorithms. We show that many sigma-point methods can be interpreted as Gaussian quadrature based methods with suitably selected covariance functions. We show that this interpretation also extends to more general multivariate Gauss--Hermite integration methods and related spherical cubature rules. Additionally, we discuss different criteria for selecting the sigma-point locations: exactness for multivariate polynomials up to a given order, minimum average error, and quasi-random point sets. The performance of the different methods is tested in numerical experiments.
\end{abstract}

\IEEEpeerreviewmaketitle

\section{Introduction}
Gaussian process quadratures \cite{OHagan:1978,OHagan:1991,Minka:2000,Osborne2012,Osborne:2012b,Sarkka+Hartikainen+Svensson+Sandblom:2014} are methods to numerically compute integrals of the form
\begin{equation}
  \mathcal{I}[\vg] = \int \vg(\vx) \, w(\vx) \, \diff\vx,
  \label{eq:gaussint}
\end{equation}
where $\vg ~:~ \mathbb{R}^n \mapsto \mathbb{R}^m$ is a (non-linear) integrand function and $w(\vx)$ is a given, typically positive, weight function such that $\int w(\vx) \, \diff\vx < \infty$. In Gaussian process quadratures the function $\vg(\vx)$ is approximated with a Gaussian process regressor \cite{Rasmussen+Williams:2006} and the integral is approximated with that of the Gaussian process regressor. 

Sigma-point methods \cite{Julier+Uhlmann:1995,Julier+Uhlmann+Durrant-Whyte:2000,Wan+Merwe:2001,Merwe:2004,Julier+Uhlmann:2004,Sarkka:2006,Simandl:2006,Sarkka:2008,Simandl:2009}  can be seen \cite{Sarkka:2013} as methods which approximate the above integrals via
\begin{equation}
  \int \vg(\vx) \, w(\vx) \, \diff\vx \approx \sum_i W_i \, \vg(\vx_i),
\end{equation}
where $W_i$ are some predefined weights and $\vx_i$ are the sigma-points (classically called abscissas). Typically the evaluation points and weights are selected such that when $\vg$ is a multivariate polynomial up to a certain order, the approximation is exact.

A particularly useful class of methods is obtained when the weight function is selected to be a multivariate Gaussian density $w(\vx) = \N(\vx \mid \vm,\MP)$. In the context of Gaussian process quadratures it then turns out that the integral of the Gaussian process regressor can be computed in closed form provided that the covariance function of the process is chosen to be a squared exponential \cite{Rasmussen+Williams:2006,Deisenroth:2012} (i.e., exponentiated quadratic). This kind of quadrature methods are also often referred to as Bayesian or Bayes--Hermite quadratures. They are closely related to Gauss--Hermite quadratures in the sense that as Gaussian quadratures can be seen to form a polynomial approximation to the integrand via point-evaluations, Gaussian process quadratures use a Gaussian process regression approximation instead \cite{OHagan:1978,OHagan:1991,Minka:2000}. Because Gaussian process regressors can be used to approximate much larger class of functions than polynomial approximations \cite{Rasmussen+Williams:2006}, they can be expected to perform much better also in numerical integration.

The selection of a Gaussian weight function is also particularly useful in non-linear filtering and smoothing, because the equations of non-linear Gaussian (Kalman) filters and smoothers \cite{Kushner:1967,Kushner:2011,Ito+Xiong:2000,Sarkka+Hartikainen:2010a,Sarkka:2013} consist of Gaussian integrals of the above form and linear operations on vectors and matrices. The selection of different weights and sigma-points leads to different brands of approximate filters and smoothers \cite{Sarkka:2013}. For example, the multidimensional Gaussian type of Gauss--Hermite quadrature and cubature
based filters and smoothers \cite{Ito+Xiong:2000,Wu+Hu+Wu+Hu:2006,Arasaratnam+Haykin:2009,Sarkka+Hartikainen:2010a,Arasaratnam+Haykin:2011} are based on explicit numerical integration of the Gaussian integrals. The unscented transform based methods as well as other sigma-point methods \cite{Julier+Uhlmann:1995,Julier+Uhlmann+Durrant-Whyte:2000,Wan+Merwe:2001,Merwe:2004,Julier+Uhlmann:2004,Sarkka:2006,Simandl:2006,Sarkka:2008,Simandl:2009} can also be retrospectively interpreted to belong to the class of Gaussian numerical integration based methods \cite{Wu+Hu+Wu+Hu:2006}. Conversely, Gaussian type of quadrature or cubature based methods can also be interpreted to be special cases of sigma-point methods. Furthermore, the classical Taylor series based methods \cite{Jazwinski:1970} and Stirling's interpolation based methods \cite{Norgaard+Poulsen+Ravn:2000,Simandl:2009} can be seen as ways to approximate the integrand such that the Gaussian integral becomes tractable (cf. \cite{Sarkka:2013}). The recent Fourier--Hermite series \cite{Sarmavuori+Sarkka:2012}, Hermite polynomial \cite{Sandblom+Svensson:2012} methods are also based on numerical approximation of the integrands.

The aim of this article is to present new Gaussian process quadrature based methods for non-linear filtering and smoothing, and to analyze their connection with sigma-point methods and multivariate numerical integration methods. We show that many sigma-point filtering and smoothing algorithms such as unscented Kalman filters and smoothers, cubature Kalman filters and smoothers, and Gauss--Hermite Kalman filters and smoothers can be seen as special cases of the proposed methods with suitably chosen covariance functions. More generally, we show that many classical multivariate Gaussian quadrature methods, including Gauss--Hermite rules \cite{Davis+Rabinowitz:1984}, and symmetric integration formulas \cite{McNamee+Stenger:1967} are special cases of the present methodology. We also discuss different criteria for selecting the sigma-point (abscissa) locations: exactness for multivariate polynomials up to a given order, minimum average error, and quasi-random point sets. 

This article is an extended version of the conference article \cite{Sarkka+Hartikainen+Svensson+Sandblom:2014}, where we analyzed the use of Gaussian process quadratures in non-linear filtering and smoothing as well as their connection to the unscented transform and Gauss--Hermite quadratures. In this article, we deepen and sharpen the analysis of those connections and extend our analysis to a more general class of spherically symmetric integration rules. We also analyze different sigma-point selection schemes as well as provide more extensive set of numerical experiments.

\section{Background}

\subsection{Non-Linear Gaussian (Kalman) Filtering and Smoothing}
Non-linear Gaussian (Kalman) filters and smoothers  \cite{Ito+Xiong:2000,Wu+Hu+Wu+Hu:2006,Sarkka+Hartikainen:2010a,Sarkka:2013} are methods which can be used to approximate the filtering distributions $p(\vx_k \mid \vy_1,\ldots,\vy_k)$ and smoothing distributions $p(\vx_k \mid \vy_1,\ldots,\vy_T)$ of non-linear state-space models of the form
\begin{equation}
\begin{split}
  \vx_k &= \vf(\vx_{k-1}) + \vq_{k-1}, \\
  \vy_k &= \vh(\vx_k) + \vr_k,
\end{split}
\label{eq:ssmodel}
\end{equation}
where, for $k=1,2,\ldots,T$, $\vx_k \in \mathbb{R}^n$ are the hidden states, $\vy_k \in \mathbb{R}^d$ are the measurements, and $\vq_{k-1} \sim \N(\mathbf{0},\MQ_{k-1})$ and $\vr_k \sim \N(\mathbf{0},\MR_k)$ are the process and measurements noises, respectively. The non-linear function $\vf(\cdot)$ is used to model the dynamics of the system and $\vh(\cdot)$ models the mapping from the states to the measurements.

Non-linear Gaussian filters (see, e.g., \cite{Sarkka:2013}, page 98) are general methods to produce Gaussian approximations to the filtering distributions:
\begin{equation}
  p(\vx_k \mid \vy_1,\ldots,\vy_k) \approx
  \N(\vx_k \mid \vm_k,\MP_k), \quad k=1,2,\ldots,T.
\end{equation}
Non-linear Gaussian smoothers (see, e.g., \cite{Sarkka:2013}, page 154) are the corresponding methods to produce approximations to the smoothing distributions:
\begin{equation}
  p(\vx_k \mid \vy_1,\ldots,\vy_T) \approx
  \N(\vx_k \mid \vm^s_k,\MP^s_k), \quad k=1,2,\ldots,T.
\end{equation}
Both Gaussian filters and smoothers above can be easily generalized to state-space models with non-additive noises (see \cite{Sarkka:2013}), but here we only consider the additive noise case.

A general additive noise one-step moment-matching-based Gaussian filter algorithm can be written in the following form.

\begin{algorithm}[Non-Linear Gaussian filter] \label{alg:gf} The prediction and update steps of the non-linear Gaussian (Kalman) filter are \cite{Ito+Xiong:2000,Sarkka:2013}:
  \begin{itemize}
  \item Prediction:
  \begin{equation}
  \begin{split}
    \vm^-_{k} &=
      \int \vf(\vx_{k-1}) \,
        \N(\vx_{k-1} \mid \vm_{k-1},\MP_{k-1}) \, \diff\vx_{k-1} \\
    \MP^-_{k} &= \int (\vf(\vx_{k-1}) - \vm^-_{k}) \,
       (\vf(\vx_{k-1}) - \vm^-_{k})^\T \\
       &\qquad \times
       \N(\vx_{k-1} \mid \vm_{k-1},\MP_{k-1}) \, \diff\vx_{k-1} 
       + \MQ_{k-1}.
  \end{split}
  \label{eq:gf_predict}
  \end{equation}

  \item Update:
  \begin{equation}
  \begin{split}    
    \vmu_{k} &= 
       \int \vh(\vx_{k}) \,
        \N(\vx_{k} \mid \vm^-_{k},\MP^-_{k}) \, \diff\vx_k \\
    \MS_{k} &= \int (\vh(\vx_{k}) - \vmu_{k}) \,
       (\vh(\vx_{k}) - \vmu_{k})^\T \\ &\qquad \times
       \N(\vx_{k} \mid \vm^-_{k},\MP^-_{k}) \, \diff\vx_{k}
       + \MR_{k} \\
    \MC_k &=  \int (\vx_k - \vm^-_k) \,
                  (\vh(\vx_k) - \vmu_k)^\T \\ &\qquad \times
                  \N(\vx_k \mid \vm^-_k,\MP^-_k) \, \diff\vx_k \\
    \MK_{k} &= \MC_k \, \MS^{-1}_{k} \\
    \vm_{k} &= \vm^-_{k} + \MK_{k} \, (\vy_k - \vmu_k) \\
    \MP_{k} &= \MP^-_{k} - \MK_{k} \, \MS_{k} \, \MK^\T_{k}.
  \end{split}
  \label{eq:gf_update}
  \end{equation}
  \end{itemize}
The filtering is started from initial mean and covariance, $\vm_0$ and $\MP_0$, respectively, such that $\vx_0 \sim \N(\vm_0,\MP_0)$. Then the prediction and update steps are applied for $k=1,2,3,\ldots,T$.
\end{algorithm}
The result of the filter is a sequence of approximations
\begin{equation}
  p(\vx_k \mid \vy_1,\ldots,\vy_k) \approx
  \N(\vx_k \mid \vm_k,\MP_k), \quad k=1,2,\ldots,T.
\end{equation}
The corresponding smoothing algorithm can be written in the following form.
\begin{algorithm}[Non-Linear Gaussian RTS smoother] \label{alg:gs}
 The equations of the non-linear Gaussian (Rauch--Tung--Striebel, RTS) smoother are the following \cite{Sarkka+Hartikainen:2010a,Sarkka:2013}:
\begin{equation}
\begin{split}
   \vm^-_{k+1} &= \int \vf(\vx_k) \,
      \mathrm{N}(\vx_k \mid  \vm_{k},\MP_{k}) \, \diff\vx_k \\
   \MP^-_{k+1} &= \int [\vf(\vx_k) - \vm^-_{k+1}] \,
      [\vf(\vx_k) - \vm^-_{k+1}]^\T \,
     \\ &\qquad \times
      \mathrm{N}(\vx_k \mid  \vm_{k},\MP_{k}) \,
      \diff\vx_k + \MQ_k \\
   \MD_{k+1} &= \int [\vx_k - \vm_{k}] \, [\vf(\vx_k) - \vm^-_{k+1}]^\T
     \\ &\qquad \times
      \mathrm{N}(\vx_k \mid  \vm_{k},\MP_{k}) \, \diff\vx_k \\
   \MG_{k} &= \MD_{k+1} \, [\MP^-_{k+1}]^{-1} \\
   \vm^s_{k} &= \vm_{k}
     + \MG_k \, (\vm^s_{k+1} - \vm^-_{k+1}) \\
   \MP^s_{k} &= \MP_{k}
     + \MG_k \, (\MP^s_{k+1} - \MP^-_{k+1}) \, \MG_k^\T.
\end{split}
\label{eq:gs}
\end{equation}
The smoothing recursion is started from the filtering result of the last time step $k = T$, that is, $\vm^s_T = \vm_T$, $\MP^s_T = \MP_T$ and proceeded backwards for $k=T-1,T-2,\ldots,1$. 
\end{algorithm}
The approximations produced by the smoother are
\begin{equation}
  p(\vx_k \mid \vy_1,\ldots,\vy_T) \approx
  \N(\vx_k \mid \vm^s_k,\MP^s_k), \quad k=1,2,\ldots,T.
\end{equation}

Both the filter and smoother above can be derived from the following Gaussian moment matching "transform" \cite{Sarkka:2013} (the terminology comes from unscented transform).

\begin{algorithm}[Gaussian moment matching of an additive transform]
  \label{alg:moment_trans1} The moment matching based Gaussian
  approximation to the joint distribution of $\vx$ and the
  transformed random variable $\vy = \vg(\vx) + \vq$, where
  $\vx \sim \N(\vm,\MP)$ and $\vq \sim
  \N(\mathbf{0},\MQ)$, is given by
  \begin{equation}
     \begin{pmatrix}
       \vx \\ \vy
     \end{pmatrix} \sim
     \N\left(
     \begin{pmatrix}
       \vm \\ \vmu_M
     \end{pmatrix},
     \begin{pmatrix}
       \MP   & \MC_M \\
       \MC_M^\T & \MS_M
     \end{pmatrix} \right),
  \end{equation}
  where
  \begin{equation}
  \begin{split}
    \vmu_M &= 
     \int \vg(\vx) \, \N(\vx \mid \vm,\MP) \, \diff\vx, \\
    \MS_M &=  \int (\vg(\vx) - \vmu_M) \,
                       (\vg(\vx) - \vmu_M)^\T \,
                       \N(\vx \mid \vm,\MP) \, \diff\vx + \MQ, \\
    \MC_M &=  \int (\vx - \vm) \,
                       (\vg(\vx) - \vmu_M)^\T \,
                       \N(\vx \mid \vm,\MP) \, \diff\vx.
  \end{split}
  \end{equation}
\end{algorithm}

\subsection{Gaussian Integration and Sigma-Point Methods}
Sigma-point filtering and smoothing methods can generally be described as methods which approximate the Gaussian integrals in the Gaussian filtering and smoothing equations (and in the Gaussian moment matching transform) as
\begin{equation}
  \int \vg(\vx) \, \N(\vx \mid \vm,\MP) \, \diff\vx
  \approx \sum_i W_i \, \vg(\vx_i),
\end{equation}
where $W_i$ are some predefined weights and $\vx_i$ are the sigma-points. Typically, the sigma-point methods use so called {\em stochastic decoupling} which refers to the idea that we do a change of variables
\begin{equation}
\begin{split}
  \int \vg(\vx) \, \N(\vx \mid \vm,\MP) \, \diff\vx 
  &= \int
  \underbrace{\vg(\vm + \sqrt{\MP} \, \vxi)}_{\tilde{\vg}(\vxi)}
     \, \N(\vxi \mid \mathbf{0},\MI) \, \diff\vxi \\
\end{split}
\label{eq:stocdec}
\end{equation}
where $\MP = \sqrt{\MP} \, \sqrt{\MP}^\T$. This implies that we only need to design weights $W_i$ and unit sigma-points $\vxi_i$ for integrating against unit Gaussian distributions:
\begin{equation}
  \int \tilde{\vg}(\vxi) \,
  \N(\vxi \mid \mathbf{0},\MI) \, \diff\vxi
  \approx \sum_i W_i \, \tilde{\vg}(\vxi_i),
\label{eq:unit_int}
\end{equation}
thus leading to approximations of the form
\begin{equation}
  \int \vg(\vx) \, \N(\vx \mid \vm,\MP) \, \diff\vx
  \approx \sum_i W_i \,
  \vg(\vm + \sqrt{\MP} \, \vxi_i).
\end{equation}

Different sigma-point methods correspond to different choices of weights $W_i$ and  unit sigma-points $\vxi_i$. For example, the canonical unscented transform \cite{Julier+Uhlmann:1995} uses the following set of $2n+1$ weights (recall that $n$ is the dimensionality of the state) and sigma-points:
\begin{equation}
\begin{split}
   W_0 &= \frac{\kappa}{n + \kappa}, \quad
   W_i = \frac{1}{2 (n + \kappa)}, \quad i=1,\ldots,2n, \\
  \vxi_i &= \begin{cases}
     \mathbf{0}, & i = 0, \\
     \sqrt{n + \kappa} \, \ve_i,& i = 1,\ldots,n, \\
    -\sqrt{n + \kappa} \, \ve_{i-n},& i = n+1,\ldots,2n. \\
  \end{cases}
\end{split}
\label{eq:utwxi}
\end{equation}
where $\kappa$ is a design parameter in the algorithm and $\ve_i \in \mathbb{R}^n$ is the unit vector towards the direction of the $i$th coordinate axis.

Note that sigma-point methods sometimes use different weights for the integrals appearing in the mean and covariance computations of Gaussian filters and smoothers. However, here we will only concentrate on the methods which use the same weights for both in order to derive more direct connections between the methods. For example, the above unscented transform weights are just a special case of more general unscented transforms (see, e.g., \cite{Sarkka:2013}).

\subsection{Gaussian Process Regression}
Gaussian process quadrature \cite{OHagan:1991,Minka:2000} is based on forming a Gaussian process (GP) regression \cite{Rasmussen+Williams:2006} approximation to the integrand using pointwise evaluations and then integrating the approximation. In GP regression \cite{Rasmussen+Williams:2006} the purpose is to predict the value of an unknown function 
\begin{equation}
  o = g(\vx)
\end{equation}
at a certain test point $(o^*,\vx^*)$ based on a finite number of
training samples $\mathcal{D} = \{ (o_j,\vx_j) : j=1,\ldots,N \}$ observed from it. The difference to classical regression is that instead of postulating a parametric regression function $g_\theta(\vx;\vtheta)$, where $\vtheta \in \mathbb{R}^D$ are the parameters, in GP regression we put a Gaussian process prior with a given covariance function $K(\vx,\vx')$ on the unknown functions $g_K(\vx)$. 

In practice, the observations are often assumed to contain noise and hence a typical model setting is:
\begin{equation}
\begin{split}
  g_K &\sim \mathrm{GP}(0,K(\vx,\vx')) \\
  o_j &= g_K(\vx_j) + \epsilon_j, \quad \epsilon_j \sim \mathrm{N}(0,\sigma^2),
\end{split}
\end{equation}
where the first line above means that the random function $g_K$ has a zero mean Gaussian process prior with the given covariance function $K(\vx,\vx')$. A commonly used covariance function is the exponentiated quadratic (also called squared exponential) covariance function
\begin{equation}
  K(\vx,\vx') = s^2 \, \exp\left(-\frac{1}{2 \ell^2}
    \| \vx - \vx' \|^2 \right),
\label{eq:secov}
\end{equation}
where $s, \ell > 0$ are parameters of the covariance function (see \cite{Rasmussen+Williams:2006}).

The GP regression equations can be derived as follows. Assume that we want to estimate the value of the noise-free function $g(\vx)$ based on its Gaussian process approximation $g_K(\vx)$
at a test point $\vx$ given the vector of observed values $\vo = ( o_1, \ldots, o_N )$. Due to the Gaussian process assumption we now get
\begin{equation}
  \begin{pmatrix}
    \vo \\ g_K(\vx)
  \end{pmatrix} \sim
  \N\left(
    \begin{pmatrix}
       \mathbf{0} \\ 0
    \end{pmatrix},
    \begin{pmatrix}
       \MK + \sigma^2 \MI & \vk(\vx) \\
       \vk^\T(\vx) & K(\vx,\vx)
    \end{pmatrix}
  \right)
\end{equation}
where $\MK = [K(\vx_i,\vx_j)]$ is the joint
covariance of observed points, $K(\vx,\vx)$ is the (co)variance of the test point, $\vk(\vx) = [K(\vx,\vx_{i})]$ is the vector cross covariances with the test point.

The Bayesian estimate of the unknown value of $g_K(\vx)$ is now given by its posterior mean, given the training data. Because everything is Gaussian, the posterior distribution is Gaussian and hence described by the posterior mean and (auto)covariance functions:
\begin{equation}
\begin{split}
  \E[g_K(\vx) \mid \vo] &= \vk^\T(\vx) \,
  (\MK + \sigma^2 \MI)^{-1} \, \vo \\
  \Cov[g_K(\vx)  \mid \vo] &= K(\vx,\vx')
  -
  \vk^\T(\vx) \, (\MK + \sigma^2 \MI)^{-1}
    \vk(\vx').
\end{split}
\end{equation}
These are the Gaussian process regression equations in their typical form \cite{Rasmussen+Williams:2006}, in the special case where $g$ is scalar. The extension to multiple output dimensions is conceptually straightforward (see, e.g., \cite{Rasmussen+Williams:2006,Sarkka:2011}), but construction of the covariance functions as well as the practical computational methods tend to be complicated 
\cite{Alvarez+Lawrence:2009,Alvarez+Luengo+Titsias+Lawrence:2010}. However, a typical easy approach to the multivariate case is to treat each of the dimensions independently.

\subsection{Gaussian Process Quadrature} \label{sec:qpquad}
In Gaussian process quadrature \cite{OHagan:1991,Minka:2000} the basic idea is to approximate the integral of a given function $g$ against a weight function $w(\vx)$, that is,
\begin{equation}
  \mathcal{I}[g] = \int g(\vx) \, w(\vx) \, \diff\vx,
  \label{eq:gaussint_b}
\end{equation}
by evaluating the function $g$ at a finite number of points and then by forming a Gaussian process approximation $g_K$ to the function. The integral is then approximated by integrating the Gaussian process approximation (or its posterior mean) which is conditioned on the evaluation points instead of the function itself. Here we assume that $g$ is scalar for simplicity as we can always take a vector function elementwise.

Gaussian process quadratures are related to a regression interpretation of classical Gaussian quadrature integration, that is, we can interpret these integration methods as orthogonal polynomial approximations of the integrand evaluated at certain finite number of points \cite{Minka:2000}. The integral is then approximated by integrating the polynomial instead of the original function. However, the aim of Gaussian process quadrature is to get a good performance in average, whereas in classical polynomial quadratures the integration rule is designed to be exact for a limited class of (polynomial) functions. Still, these approaches are very much linked together \cite{Minka:2000}.

Due to linearity of integration, the posterior mean of the integral of the Gaussian process regressor is given as
\begin{equation}
\begin{split}
    \E\left[\int g_K(\vx) \, w(\vx) \, \diff \vx \mid \vo \right] &= 
  \int \E\left[g_K(\vx) \mid \vo \right] \, w(\vx) \, \diff \vx,
\end{split}
\end{equation}
where the ``training set'' $\vo = \begin{pmatrix} g(\vx_1), \ldots, g(\vx_N) \end{pmatrix}$ now contains the values of the function $g$ evaluated at certain selected inputs.

The posterior variance of the integral can be evaluated in an analogous manner, and it is sometimes used to optimize the evaluation points of the function $g_N$ \cite{OHagan:1991,Minka:2000,Osborne2012,Osborne:2012b}. The posterior covariance of the approximation is
\begin{equation}
\begin{split}
   &\Var\left[\int g_K(\vx) \, w(\vx) \, \diff \vx \mid \vo \right] \\ & \qquad = 
  \iint \Cov\left[g_K(\vx) \mid \vo \right] \, w(\vx) \, \diff \vx \, w(\vx') \, \diff \vx'.
\end{split}
\end{equation}
That is, when we approximate the integral~\eqref{eq:gaussint_b} with the posterior mean we have
\begin{equation}
\begin{split}
  \int g(\vx) \, w(\vx) \, \diff\vx
 \approx \left[ \int \vk^\T(\vx) \, w(\vx) \,
  \diff\vx \, \right]
  (\MK + \sigma^2 \MI)^{-1} \, \vo,
\end{split}
\label{eq:gpquad0}
\end{equation}
The posterior variance of the (scalar) integral is
\begin{equation}
\begin{split}
   &\Var\left[\int g_K(\vx) \, w(\vx) \, \diff \vx \mid \vo \right] \\ & \qquad = 
  \iint K(\vx,\vx') \, w(\vx) \, \diff \vx \, w(\vx') \, \diff \vx' \\
  &- \left[ \int \vk^\T(\vx) \, w(\vx) \,
  \diff\vx \, \right]
  (\MK + \sigma^2 \MI)^{-1} \,
   \left[ \int \vk(\vx') \, w(\vx') \,
  \diff\vx' \, \right]
\end{split}
\label{eq:gpquadvar0}
\end{equation}
In this article we are specifically interested in the case of Gaussian weight function, which then reduces the integral appearing in the above expressions \eqref{eq:gpquad0} and \eqref{eq:gpquadvar0} to
\begin{equation}
\begin{split}
  \left[\int \vk^\T(\vx) \, \, w(\vx) \, \diff\vx \right]_i
  = \int K(\vx,\vx_i)  \N(\vx \mid \vm,\MP) \, \diff\vx
\end{split}
\end{equation}
It is now easy to see that when the covariance function is a squared exponential $K(\vx,\vx_i) = s^2 \, \exp(-(2 \ell^2)^{-1} \, \| \vx - \vx_i \|^2)$, this integral can be easily computed in closed form by using the computation rules for Gaussian distributions. Furthermore if the covariance function is a multivariate polynomial, then these integrals are given by the moments of the Gaussian distributions, which are also available in closed form.

\section{Gaussian process quadratures for sigma-point filtering and smoothing}
In this section we start by showing how Gaussian process quadratures (GPQ) can be seen as sigma-point methods and then introduce the Gaussian process transform. The Gaussian process transform then enables us to construct GPQ-based non-linear filters and smoothers analogously to \cite{Sarkka:2013}. 

\subsection{GPQ as a sigma-point method}
In this section the aim is to shown how Gaussian process quadratures (GPQ) can be seen as sigma-point methods.
\begin{lemma}[GPQ as a sigma-point method]
The Gaussian process quadrature (or Bayes--Hermite/Bayesian quadrature) can be seen is a sigma-point-type of integral approximation
\begin{equation}
  \int \vg(\vx) \, \N(\vx \mid \vm,\MP) \, \diff\vx
  \approx \sum_{i=1}^N W_i \, \vg(\vx_i),
\end{equation}
where $\vx_i = \vm + \sqrt{\MP} \, \vxi_i$ with the unit sigma-points $\vxi_i$ are selected according to a predefined criterion, and the weights are determined by
\begin{equation}
\begin{split}
  W_i &=
  \bigg[ \left( \int \vk^\T(\vxi) \, \mathrm{N}(\vxi \mid \mathbf{0},\MI)
  \, \diff\vxi \right) \,
   (\MK + \sigma^2 \MI)^{-1} \bigg]_i,
\end{split}
\label{eq:gp_wi}
\end{equation}
where $\MK = [K(\vxi_i,\vxi_j)]$ is the matrix of unit sigma-point covariances and $\vk(\vxi) = [K(\vxi,\vxi_{i})]$ is the vector cross covariances. In principle, the choice of unit sigma-points above is completely free, but good choices of them are discussed in the following sections. 
\end{lemma}

\begin{proof}
Let us first use stochastic decoupling \eqref{eq:stocdec} which enables us to only consider unit-Gaussian integration formulas of the form \eqref{eq:unit_int}. Because we can integrate vector functions element-by-element, without loss of generality we can assume that $g(\vx)$ is single-dimensional. Let us now model the function $\vxi \mapsto g(\vm + \sqrt{\MP} \, \vxi)$ as a Gaussian process $g_K$ with a given covariance function $K(\vxi,\vxi')$ and fix the training set for the GP regressor by selecting the points $\vxi_{i}$, $i=1,\ldots,N$, which also determines the corresponding points $\vx_i = \vm + \sqrt{\MP} \, \vxi_i$ such that the training set is $\vo = \begin{pmatrix} g(\vx_1), \ldots, g(\vx_N) \end{pmatrix}$. The GP approximation to the integral now follows from \eqref{eq:gpquad0}:
\begin{equation}
\begin{split}
  &\int g(\vm + \sqrt{\MP} \, \vxi) \, \N(\vxi \mid 0,\MI) \, \diff\vxi \\
  & \approx \left[ \int \vk^\T(\vxi) \, \N(\vxi \mid 0,\MI) \,
  \diff\vxi \, \right]
  (\MK + \sigma^2 \MI)^{-1} \, \vo,
\end{split}
\end{equation}
which when simplified and applied to all the dimensions of $\vg$ gives the result.
\end{proof}
Note that above we actually assume that the stochastically-decoupled-function $\vxi \mapsto g(\vm + \sqrt{\MP} \, \vxi)$ instead of the original integrand $g(\vx)$ has the given covariance function. The reason for this modeling choice is that it enables us to decouple the mean and covariance from the integration formula and hence is computationally beneficial. This also makes the result invariant to affine transformations of the state and it also has a property that the variability of the functions corresponds to the scale of the problem. However, on the other hand, one might argue that it is the function $g(\vx)$ which should actually model and using the stochastically-decoupled-function is ``wrong''. 

\begin{remark}[Variance of GPQ]
  From Equation~\eqref{eq:gpquadvar0} we get that the component-wise variances of the Gaussian process quadrature approximation can be expressed as
\begin{equation}
\begin{split}
  &V_j = 
  \iint K(\vxi,\vxi') \, \mathrm{N}(\vxi \mid \mathbf{0},\MI)
  \, \diff\vxi\, \mathrm{N}(\vxi' \mid \mathbf{0},\MI)
  \, \diff\vxi' \\
  &-
  \int \vk^\T(\vxi) \, \mathrm{N}(\vxi \mid \mathbf{0},\MI)
  \, \diff\vxi
  \,
   (\MK + \sigma^2 \MI)^{-1}
  \int \vk(\vxi') \, \mathrm{N}(\vxi' \mid \mathbf{0},\MI)
  \, \diff\vxi'
\end{split}
\label{eq:gpquadvar1}
\end{equation}
\end{remark}
Using the above integration approximations we can also define a general Gaussian process transform as follows. The reason for introducing the transform is that the corresponding approximate filters and smoothers can be readily constructed in terms of the transform (cf. \cite{Sarkka:2013}), which we will do in the next section.
\begin{algorithm}[Gaussian process transform] \label{alg:gptrans}
  The Gaussian process quadrature based Gaussian approximation to the joint distribution of $\vx$ and the transformed random variable $\vy = \vg(\vx) + \vq$, where $\vx \sim \N(\vm,\MP)$ and $\vq \sim \N(\mathbf{0},\MQ)$, is given by
  \begin{equation}
     \begin{pmatrix}
       \vx \\ \vy
     \end{pmatrix} \sim
     \N\left(
     \begin{pmatrix}
       \vm \\ \vmu_{\mathrm{GP}}
     \end{pmatrix},
     \begin{pmatrix}
       \MP   & \MC_{\mathrm{GP}} \\
       \MC_{\mathrm{GP}}^\T & \MS_{\mathrm{GP}}
     \end{pmatrix} \right),
  \end{equation}
  where
  \begin{equation}
  \begin{split}
   \vx_i &= \vm + \sqrt{\MP} \, \vxi_i \\
    \vmu_{\mathrm{GP}} &= 
    \sum_{i=1}^N W_i \, \vg(\vx_i), \\
    \MS_{\mathrm{GP}} &= 
     \sum_{i=1}^N W_i \, (\vg(\vx_i) - \vmu_{\mathrm{GP}}) \,
                         (\vg(\vx_i) - \vmu_{\mathrm{GP}})^\T
                         + \MQ, \\
    \MC_{\mathrm{GP}} &= 
     \sum_{i=1}^N W_i \, (\vx_i - \vm) \,
                         (\vg(\vx_i) - \vmu_{\mathrm{GP}})^\T,
  \end{split}
  \end{equation}
where $\vxi_i$ is some fixed set of sigma/training points and the weights are given by Equation~\eqref{eq:gp_wi} with some selected covariance function $K(\vxi,\vxi')$.
\end{algorithm}

In this article, at least in the analytical results, we usually assume that the measurements are noise-free, that is, $\sigma^2 = 0$. This enables us to obtain analytically exact relationships with the classical quadrature methods. However, when using Gaussian process quadratures as numerical integration method, it is often beneficial to have at least a small non-zero value for $\sigma^2$ in \eqref{eq:gp_wi}. This kind of ``jitter'' stabilizes numerics and can even be sometimes used to compensate for inaccuracies in modeling.

\begin{example}[GPT with squared exponential kernel] \label{ex:gpt_se}
Let us now consider $\vxi \in \mathbb{R}$ and select the sigma-point locations to be the ones of unscented transform \eqref{eq:utwxi}. With the squared exponential covariance function \eqref{eq:secov} and noise-free measurements ($\sigma^2 = 0$) we then get the weights:
\begin{equation}
W_{0:2} =
\begin{pmatrix} \frac{\mathrm{e}^{-\frac{\kappa + 1}{2\, \left({\ell}^2 + 1\right)}}\, \left(\ell\, \mathrm{e}^{\frac{\kappa + 1}{2\, \left({\ell}^2 + 1\right)}} - 2\, \ell\, \mathrm{e}^{\frac{3\, \left(\kappa + 1\right)}{2\, {\ell}^2}} + \ell\, \mathrm{e}^{\frac{\kappa + 1}{2\, \left({\ell}^2 + 1\right)}}\, \mathrm{e}^{\frac{2\, \left(\kappa + 1\right)}{{\ell}^2}}\right)}{\sqrt{{\ell}^2 + 1}\, {\left(\mathrm{e}^{\frac{\kappa + 1}{{\ell}^2}} - 1\right)}^2}\\ -\frac{\ell\, \mathrm{e}^{\frac{\left(2\, {\ell}^2 + 3\right)\, \left(\kappa + 1\right)}{2\, {\ell}^2\, \left({\ell}^2 + 1\right)}}\, \left(\mathrm{e}^{\frac{\kappa + 1}{2\, \left({\ell}^2 + 1\right)}} - \mathrm{e}^{\frac{\kappa + 1}{2\, {\ell}^2}}\right)}{\sqrt{{\ell}^2 + 1}\, {\left(\mathrm{e}^{\frac{\kappa + 1}{{\ell}^2}} - 1\right)}^2}\\ -\frac{\ell\, \mathrm{e}^{\frac{\left(2\, {\ell}^2 + 3\right)\, \left(\kappa + 1\right)}{2\, {\ell}^2\, \left({\ell}^2 + 1\right)}}\, \left(\mathrm{e}^{\frac{\kappa + 1}{2\, \left({\ell}^2 + 1\right)}} - \mathrm{e}^{\frac{\kappa + 1}{2\, {\ell}^2}}\right)}{\sqrt{{\ell}^2 + 1}\, {\left(\mathrm{e}^{\frac{\kappa + 1}{{\ell}^2}} - 1\right)}^2} \end{pmatrix}.
\end{equation}
An interesting property is that in the limit $\ell \to \infty$ we get
\begin{equation}
\lim_{\ell \to \infty} W_{0:2} =
\left(\begin{array}{c} \frac{\kappa}{\kappa + 1}\\ \frac{1}{2\, \left(\kappa + 1\right)}\\ \frac{1}{2\, \left(\kappa + 1\right)} \end{array}\right)
\end{equation}
which are the unscented transform weights. We return to this relationship in Section~\ref{sec:minka}.
\end{example}

\subsection{GPQs in filtering and smoothing}

In this section we show how to construct filters and smoothers using the Gaussian process quadrature approximations. Because Algorithm~\ref{alg:gptrans} can be seen as a sigma-point method, analogously to other sigma-point filters considered, for example, in \cite{Sarkka:2013}, we can now formulate the following sigma-point filter for model \eqref{eq:ssmodel}, which uses the the unit sigma-points $\vxi_i$ and weights $W_i$ defined by Algorithm~\ref{alg:gptrans}. 
\begin{algorithm}[Gaussian process quadrature filter] The filtering is started from initial mean and covariance, $\vm_0$ and $\MP_0$, respectively, such that $\vx_0 \sim \N(\vm_0,\MP_0)$. Then the following prediction and update steps are applied for $k=1,2,3,\ldots,T$.

\noindent Prediction:

\begin{enumerate}
\item Form the sigma points as follows:
$\mathcal{X}^{(i)}_{k-1} = \vm_{k-1} + \sqrt{\MP_{k-1}} \, \vxi_i, i=1,\ldots,N$.

\item Propagate the sigma points through the dynamic model:
    $\hat{\mathcal{X}}^{(i)}_{k} = \vf(\mathcal{X}^{(i)}_{k-1}),
    i = 1,\ldots,N$.

\item Compute the predicted mean $\vm^-_{k}$ and the predicted
  covariance $\MP^-_{k}$:
  \begin{equation}
  \begin{split}
    \vm^-_{k} &= \sum_{i=1}^{N} W_i \, \hat{\mathcal{X}}^{(i)}_{k}, \\
    \MP^-_{k} &= \sum_{i=1}^{N} W_i \, 
        (\hat{\mathcal{X}}^{(i)}_{k} - \vm^-_{k}) \,
        (\hat{\mathcal{X}}^{(i)}_{k} - \vm^-_{k})^\T + \MQ_{k-1}.
  \end{split}
  \nonumber
  \end{equation}
\end{enumerate}

\noindent Update:

\begin{enumerate}
\item Form the sigma points:
  $\mathcal{X}^{-(i)}_{k} = \vm^-_{k} + 
        \sqrt{\MP^-_{k}} \, \vxi_i, 
   i=1,\ldots,N$.

\item Propagate sigma points through the measurement model:
 $\hat{\mathcal{Y}}^{(i)}_{k} = \vh(\mathcal{X}^{-(i)}_{k}),
    i = 1 \ldots N$.

\item Compute the predicted mean $\vmu_k$, the predicted covariance of the measurement $\MS_k$, and the cross-covariance of the state and the measurement $\MC_k$:
  \begin{equation}
  \begin{split}
    \vmu_{k}    &= \sum_{i=1}^{N} W_i \, \hat{\mathcal{Y}}^{(i)}_{k}, \\
    \MS_{k} &= \sum_{i=1}^{N}
        W_i \, (\hat{\mathcal{Y}}^{(i)}_{k} - \vmu_{k}) \,
        (\hat{\mathcal{Y}}^{(i)}_{k} - \vmu_{k})^\T + \MR_{k}, \\
    \MC_{k} &= \sum_{i=1}^{N} W_i \, 
        (\mathcal{X}^{-(i)}_{k} - \vm^-_{k}) \,
        (\hat{\mathcal{Y}}^{(i)}_{k} - \vmu_{k})^\T.
  \end{split}
\nonumber
  \end{equation}

\item Compute the filter gain $\MK_k$ and the filtered state mean
$\vm_k$ and covariance $\MP_k$, conditional on the measurement $\vy_k$:
\begin{equation}
\begin{split}
  \MK_k &= \MC_k \, \MS_k^{-1}, \\
  \vm_k &= \vm^-_k + \MK_k \, \left[ \vy_k - \vmu_k \right], \\
  \MP_k &= \MP^-_k - \MK_k \, \MS_k \, \MK_k^\T.
\end{split}
\nonumber
\end{equation}
\end{enumerate}
\end{algorithm}
Further following the line of thought in \cite{Sarkka:2013} we can formulate a sigma-point smoother using the unit sigma-points and weights from Algorithm~\ref{alg:gptrans}.
\begin{algorithm}[Gaussian process quadrature sigma-point RTS smoother]
The smoothing recursion is started from the filtering result of the last time step $k = T$, that is, $\vm^s_T = \vm_T$, $\MP^s_T = \MP_T$ and proceeded backwards for $k=T-1,T-2,\ldots,1$ as follows. 
\begin{enumerate}
\item Form the sigma points:
   $\mathcal{X}^{(i)}_{k} = \vm_{k} + \sqrt{\MP_{k}} \, \vxi_i,
   i=1,\ldots,N$.

\item Propagate the sigma points through the dynamic model:
  $\hat{\mathcal{X}}^{(i)}_{k+1} = \vf(\mathcal{X}^{(i)}_{k}),
    i = 1,\ldots,N$.

\item Compute the predicted mean $\vm^-_{k+1}$, the predicted
  covariance $\MP^-_{k+1}$, and the cross-covariance
  $\MD_{k+1}$:
  \begin{equation}
  \begin{split}
    \vm^-_{k+1} &= \sum_{i=1}^{N} W_i \, \hat{\mathcal{X}}^{(i)}_{k+1}, \\
    \MP^-_{k+1} &= \sum_{i=1}^{N} W_i \, 
     (\hat{\mathcal{X}}^{(i)}_{k+1} - \vm^-_{k+1}) \,
     (\hat{\mathcal{X}}^{(i)}_{k+1} - \vm^-_{k+1})^\T + \MQ_{k}, \\
    \MD_{k+1} &= \sum_{i=1}^{N} W_i \, 
     (\mathcal{X}^{(i)}_{k} - \vm_k) \,
     (\hat{\mathcal{X}}^{(i)}_{k+1} - \vm^-_{k+1})^\T.
  \end{split}
\nonumber
  \end{equation}
  
\item Compute the gain $\MG_k$, mean $\vm^s_k$ and covariance
  $\MP^s_k$ as follows:
\begin{equation}
\begin{split}
    \MG_k &= \MD_{k+1} \, [\MP^-_{k+1}]^{-1}, \\
  \vm^s_k &= \vm_k + \MG_k \,
    ( \vm^s_{k+1} - \vm^-_{k+1} ), \\
  \MP^s_k &= \MP_k + \MG_k \, 
    ( \MP^s_{k+1} - \MP^-_{k+1} ) \, \MG_k^\T.
\end{split}
\nonumber
\end{equation}
\end{enumerate}
\end{algorithm}
Note that we could cope with non-additive noises in the model by using augmented forms of the above filters and smoothers as in \cite{Sarkka:2013}. The fixed-point and fixed-lag smoothers can also be derived analogously as was done in the same reference.

\section{Selection of covariance functions and sigma-point locations}

The accuracy of the Gaussian process quadrature method and hence the accuracy of the filtering and smoothing methods using it is affected by
\begin{enumerate}
\item the covariance function $K(\vxi,\vxi')$ used and
\item the sigma-point locations $\vxi_i$.
 \end{enumerate}
Once both of the above are fixed, the weights are determined by Equation~\eqref{eq:gp_wi}. In this section we discuss certain useful choices of covariance functions as well as "optimal" choices of sigma-point locations for them. We also discuss the connection of the resulting methods with sigma-point methods such as unscented transforms and Gauss--Hermite quadratures.

\subsection{Squared exponential and minimum variance point sets}
In a machine learning context \cite{Rasmussen+Williams:2006} the default choice for a covariance function of a Gaussian process is the squared exponential covariance function in Equation~\eqref{eq:secov}. What makes it convenient in Gaussian process quadrature context is that the integral required for computing the weights in Equation~\eqref{eq:gp_wi} can be evaluated in closed form (cf. \cite{Minka:2000,Deisenroth:2012}). It turns out that the posterior variance can be computed in closed form as well which is useful because for a given set of sigma-points we can immediately compute the expected error in the integral approximation (assuming that the integrand is indeed a GP) -- this is possible because the variance does not depend on the observations at all.

One way to determine the sigma-point locations is to select them to minimize the posterior variance of the integral approximation \cite{OHagan:1991,Minka:2000}. In our case this corresponds to minization of the variance in Equation~\eqref{eq:gpquadvar1} with respect to the points $\vxi_{1:N}$. Although the minimization is not possible in closed form, with a moderate $N$ this optimization can be done numerically. Figure~\ref{fig:gp_sigmas} shows examples of minimum variance point sets optimized by using the Broyden--Fletcher--Goldfarb--Shanno (BFGS) algorithm \cite{Fletcher:1987}. 

\begin{figure}[!t]
\centering
\subfloat[5 points]{\includegraphics[width=0.2\textwidth]{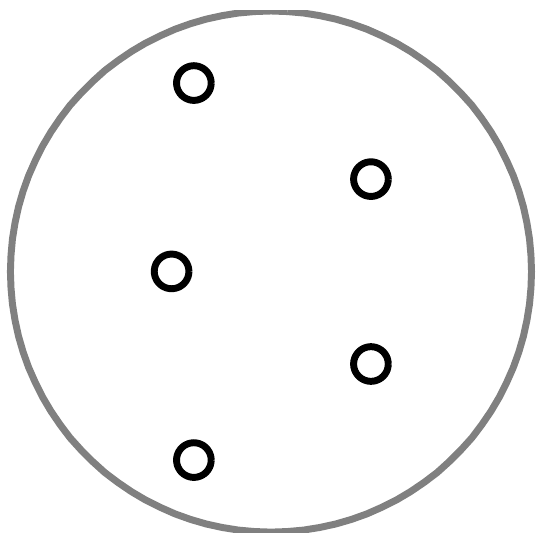}%
\label{fig:gp_sigmas_5}}
\hfil
\subfloat[10 points]{\includegraphics[width=0.2\textwidth]{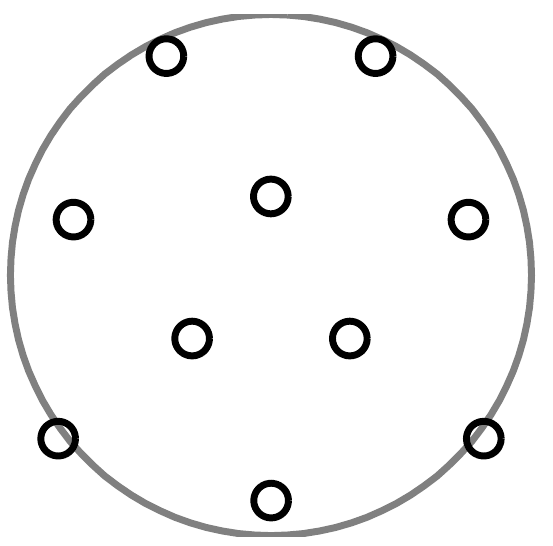}%
\label{fig:gp_sigmas_10}} \\
\subfloat[15 points]{\includegraphics[width=0.2\textwidth]{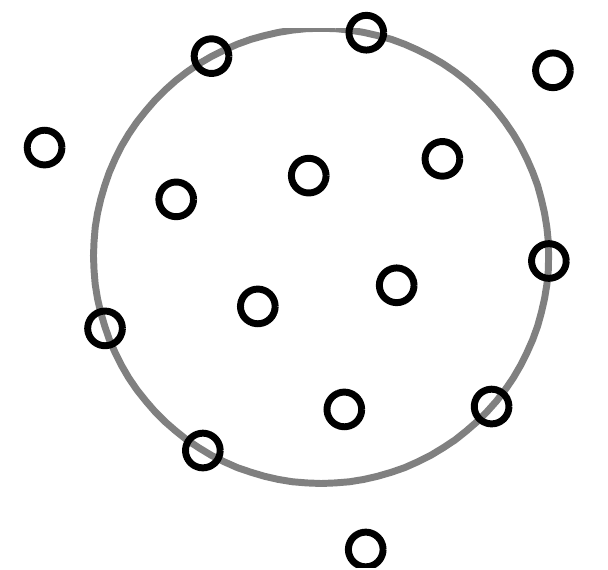}%
\label{fig:gp_sigmas_15}}
\hfil
\subfloat[20 points]{\includegraphics[width=0.2\textwidth]{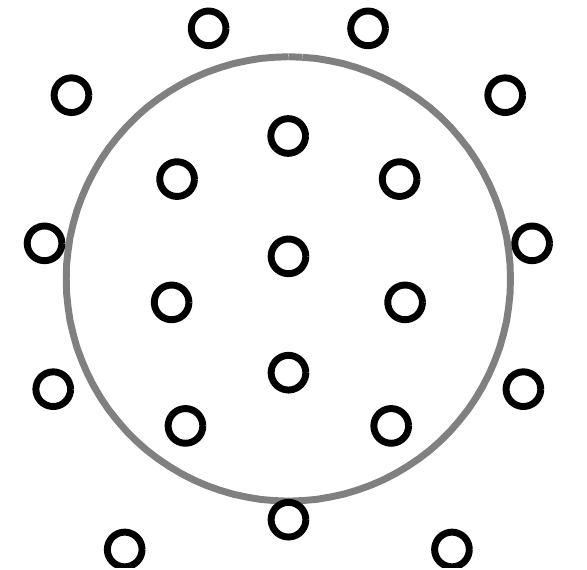}%
\label{ig:gp_sigmas_20}}
\caption{Minimum variance point sets for the squared exponential covariance function.}
\label{fig:gp_sigmas}
\end{figure}

The squared exponential covariance function is not the only possible choice for a covariance function. From the machine learning context we could, for example, choose a Mat\'ern covariance function or some of the scale-mixture-based covariance functions \cite{Rasmussen+Williams:2006}. In that case the weight integral \eqref{eq:gp_wi} becomes less trivial, but at least we always have a chance to precompute the weights using some (other) multivariate quadrature method. The sigma-point optimization could also be done similarly as for the squared exponential covariance function.

\subsection{UT and spherical cubature rules}

In addition to the squared exponential covariance funtion, another useful class of covariance function are polynomial covariance functions. They correspond to linear-in-parameters regression using polynomials as the regressor functions. It turns out that also for polynomial covariance functions we can compute the weights \eqref{eq:gp_wi} in closed form. What is even more interesting is that the Gaussian process quadratures reduce to classical numerical integration methods. In this section we show that with certain selections of symmetric evaluation points we get a classical family of spherically symmetric integration methods of McNamee and Stenger \cite{McNamee+Stenger:1967} of which the unscented transform \cite{Julier+Uhlmann:1995,Julier+Uhlmann+Durrant-Whyte:2000} can be (retrospectively) seen as a special case \cite{Julier+Uhlmann:2004}. More detailed information on the multivariate Hermite polynomials used below can be found in Appendix~\ref{sec:fh}.

\begin{theorem}[UT covariance function] \label{the:ut3_cov}
Assume that
\begin{equation}
  K(\vxi,\vxi') = \sum_{q=0}^3 \sum_{|J| = q} \sum_{p=0}^3  \sum_{|I|=p}
\frac{1}{\mathcal{I}! \, \mathcal{J}!}  \lambda_{\mathcal{I},\mathcal{J}} \, H_\mathcal{I}(\vxi) \, H_\mathcal{J}(\vxi')
\label{eq:ut3_cov}
\end{equation}
where $\lambda_{\mathcal{I},\mathcal{J}}$'s form a positive definite covariance matrix and $H_\mathcal{I}(\vxi)$ are multivariate Hermite polynomials (see Appendix~\ref{sec:fh}). If we now select the evaluation points as in UT \eqref{eq:utwxi}, then the GPQ weights $W_i$ become the UT weights. Furthermore, the posterior variance of the integral approximation is exactly zero.
\end{theorem}

\begin{proof}
The prior $g_K \sim \mathrm{GP}(0,K(\vxi,\vxi'))$ with the above covariance is equivalent to a parametric model of the form
\begin{equation}
  g_K(\vxi) = \sum_{p=0}^{3} \sum_{|\mathcal{I}| = p} 
  \frac{1}{\mathcal{I}!} c_{\mathcal{I}} \, H_{\mathcal{I}}(\vxi),
\label{eq:ut3_gf}
\end{equation}
where $c_{\mathcal{I}}$ are zero mean Gaussian random variables with the covariances $\lambda_{\mathcal{I},\mathcal{J}} = \E\left[c_{\mathcal{I}} \, c_{\mathcal{J}}\right]$. When the joint covariance matrix $\Lambda = [\lambda_{\mathcal{I},\mathcal{J}}]$ is non-singular, the posterior covariance of the integral being zero is equivalent to that the integral rule is exact for all functions of the form \eqref{eq:ut3_gf} with arbitrary coefficients. Clearly with the UT evaluations points, the UT weights are the unique ones that have this property (see, e.g., \cite{Sarkka:2013}) and hence the result follows. 
\end{proof}

Note that the above result also covers the cubature transform (CT), that is, the moment matching rule used in the cubature Kalman filter (CKF) and the smoother, because the transform is a special case of UT \cite{Sarkka:2013}.

\begin{theorem}[Higher order UT covariance function] \label{the:utn_cov}
Assume that
\begin{equation}
  K(\vxi,\vxi') = \sum_{q=0}^P \sum_{|J| = q} \sum_{p=0}^P \sum_{|I|=p}
  \frac{1}{\mathcal{I}! \, \mathcal{J}!} \lambda_{\mathcal{I},\mathcal{J}} \, H_\mathcal{I}(\vxi) \, H_\mathcal{J}(\vxi')
\label{eq:utn_cov}
\end{equation}
If we select the evaluation points according to order $P = 5,7,9,\ldots$ rules in \cite{McNamee+Stenger:1967}, we obtain the higher order integration formulas in \cite{McNamee+Stenger:1967}, which are often referred to as fifth order, seventh order, ninth order and higher order UTs.
\end{theorem}

\begin{proof}
The result follows analogously to the 3rd order case above.
\end{proof}

\begin{figure}[!t]
\centering
\subfloat[UT-3]{\includegraphics[width=0.2\textwidth]{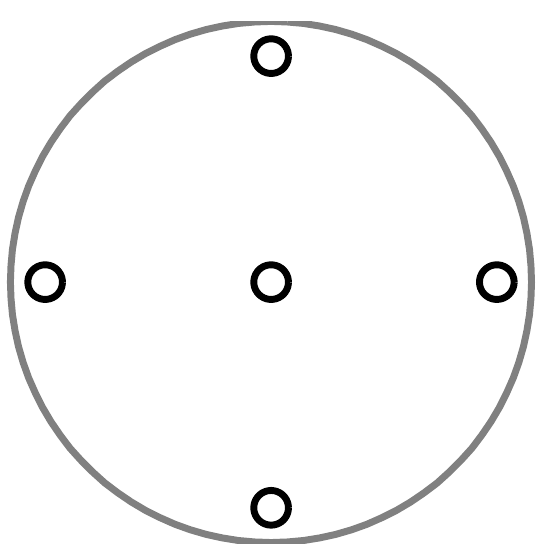}%
\label{fig:ut3_sigmas}}
\hfil
\subfloat[UT-5]{\includegraphics[width=0.2\textwidth]{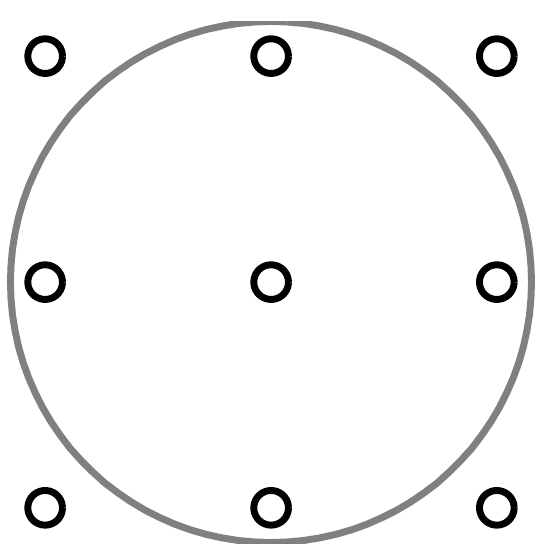}%
\label{fig:ut5_sigmas}}
\caption{Unscented transform point sets.}
\label{fig:ut_sigmas}
\end{figure}

\begin{example}[Derivation of UT weights from GPQ]
  Let $\vxi \in \mathbb{R}^2$ and consider the GPQ with UT \eqref{eq:utwxi} sigma-points and the covariance function \eqref{eq:ut3_cov}. With $\sigma = 0$ and $\lambda_{\mathcal{I},\mathcal{J}} = \delta_{\mathcal{I},\mathcal{J}}$ we then obtain the covariance matrix in \eqref{eq:gpq_cov}.
\begin{figure*}[htb]
\begin{equation}
\begin{split}
  \MK &= \begin{pmatrix} \frac{3}{2} & 1 - \frac{\kappa}{4} & 1 - \frac{\kappa}{4} & 1 - \frac{\kappa}{4} & 1 - \frac{\kappa}{4}\\ 1 - \frac{\kappa}{4} & \frac{{\kappa}^3}{36} + \frac{{\kappa}^2}{4} + \frac{13\, \kappa}{6} + \frac{91}{18} & \frac{1}{2} - \frac{\kappa}{2} &  - \frac{{\kappa}^3}{36} + \frac{{\kappa}^2}{4} - \frac{7\, \kappa}{6} - \frac{37}{18} & \frac{1}{2} - \frac{\kappa}{2}\\ 1 - \frac{\kappa}{4} & \frac{1}{2} - \frac{\kappa}{2} & \frac{{\kappa}^3}{36} + \frac{{\kappa}^2}{4} + \frac{13\, \kappa}{6} + \frac{91}{18} & \frac{1}{2} - \frac{\kappa}{2} &  - \frac{{\kappa}^3}{36} + \frac{{\kappa}^2}{4} - \frac{7\, \kappa}{6} - \frac{37}{18}\\ 1 - \frac{\kappa}{4} &  - \frac{{\kappa}^3}{36} + \frac{{\kappa}^2}{4} - \frac{7\, \kappa}{6} - \frac{37}{18} & \frac{1}{2} - \frac{\kappa}{2} & \frac{{\kappa}^3}{36} + \frac{{\kappa}^2}{4} + \frac{13\, \kappa}{6} + \frac{91}{18} & \frac{1}{2} - \frac{\kappa}{2}\\ 1 - \frac{\kappa}{4} & \frac{1}{2} - \frac{\kappa}{2} &  - \frac{{\kappa}^3}{36} + \frac{{\kappa}^2}{4} - \frac{7\, \kappa}{6} - \frac{37}{18} & \frac{1}{2} - \frac{\kappa}{2} & \frac{{\kappa}^3}{36} + \frac{{\kappa}^2}{4} + \frac{13\, \kappa}{6} + \frac{91}{18} \end{pmatrix}
\end{split}
\label{eq:gpq_cov}
\end{equation}
\end{figure*}
It also turns out that 
\begin{equation}
\begin{split}
  \int \vk^\T(\vxi) \, \mathrm{N}(\vxi \mid \mathbf{0},\MI) \, \diff\vxi = \begin{pmatrix} 1 & \cdots & 1 \end{pmatrix}
\end{split}
\end{equation}
and finally
\begin{equation}
\begin{split}
  W_{0:4} = \begin{pmatrix} \frac{\kappa}{\kappa + 2} & \frac{1}{2\, \left(\kappa + 2\right)} & \frac{1}{2\, \left(\kappa + 2\right)} & \frac{1}{2\, \left(\kappa + 2\right)} & \frac{1}{2\, \left(\kappa + 2\right)} \end{pmatrix}^{\T},
\end{split}
\end{equation}
which are indeed the UT weights.
\end{example}

\subsection{Multivariate Gauss--Hermite point sets}

The multivariate Gauss--Hermite point sets (see, e.g., \cite{Ito+Xiong:2000,Sarkka:2013}) of order $P$ are exact for monomials of of the form $x_1^{p_1} \times \cdots \times x_n^{p_n}$, where $p_i \le 2P-1$ for $i=1,\ldots,n$. This implies the following covariance function class.

\begin{theorem}[Gauss--Hermite covariance function]
Assume that
\begin{equation}
  K(\vxi,\vxi') = \sum_{\max \mathcal{J}  \le 2P-1} \sum_{\max \mathcal{I}  \le 2P-1}
  \frac{1}{\mathcal{I}! \, \mathcal{J}!}  \lambda_{\mathcal{I},\mathcal{J}} \, H_\mathcal{I}(\vxi) \, H_\mathcal{J}(\vxi')
\end{equation}
where $\lambda_{\mathcal{I},\mathcal{J}}$'s form a positive definite covariance matrix and $H_\mathcal{I}(\vxi)$ are multivariate Hermite polynomials. If we now select the evaluation points to form a cartesian product of roots of Hermite polynomials of order $P$, then the GPQ weights $W_i$ become the multivariate Gauss--Hermite quadrature weights. The posterior variance of the integral approximation is again exactly zero.
\end{theorem}

\begin{proof}
  Again the result follows from the equivalence of the polynomial approximations and polynomial covariance functions together with the uniqueness of the Gauss--Hermite rule for exact integration of this same function class.
\end{proof}

\begin{figure}[!t]
\centering
\subfloat[GH-4]{\includegraphics[width=0.2\textwidth]{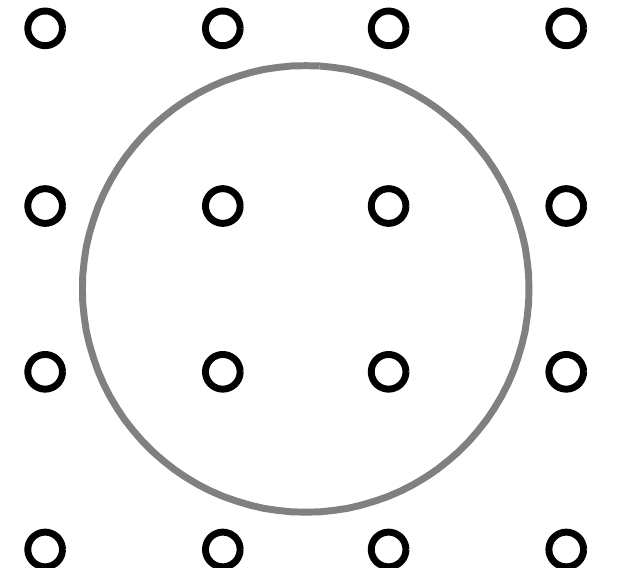}%
\label{fig:ut3_sigmas}}
\hfil
\subfloat[GH-5]{\includegraphics[width=0.2\textwidth]{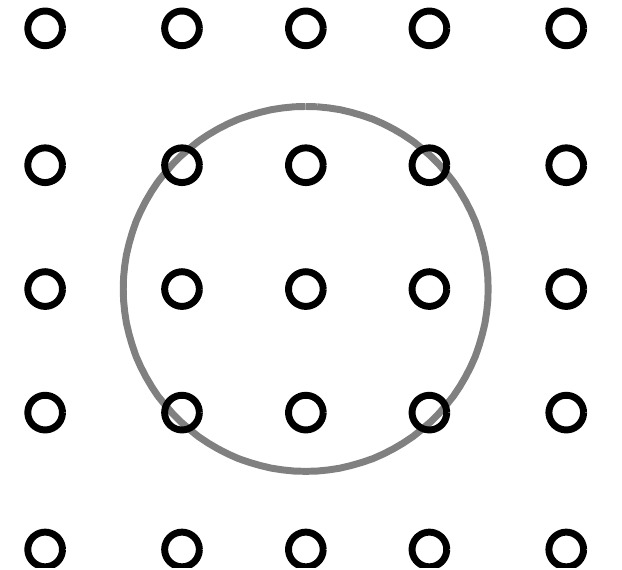}%
\label{fig:ut5_sigmas}}
\caption{Gauss--Hermite point sets.}
\label{fig:gh_sigmas}
\end{figure}

Even when we are using polynomial covariance functions, we are by no means restricted to using the specific points sets corresponding to the classical integration rules. However, obviously, given the order of the polynomial kernel and number of sigma-points they are also minimum variance points sets and hence good choices also in average -- provided that the integrand is indeed a polynomial. In any case, for an arbitrary set of sigma-points we can use Equation~\eqref{eq:gp_wi} to give the corresponding minimum variance weights. 

\subsection{Connection between squared exponential and polynomial Gaussian process quadratures} \label{sec:minka}
As discussed in \cite{Minka:2000}, the Gaussian process quadrature with squared exponential covariance function also has a strong connection with classical quadrature methods. This is because we can consider a set of damped polynomial basis functions of the form $\phi_i(\xi) = x^i \, \exp(-x^2 / (2\ell^2))$, which at least informally speaking can be seen to converge to a polynomial basis when $\ell \to \infty$. We can now consider a family of random functions (Gaussian processes) of the form
\begin{equation}
    g_\ell(x) =\sum_j c_j \, \phi_k(x) = \sum_j c_j \, x^j \, \exp\left(-\frac{x^2}{2\ell^2} \right),
\end{equation}
where $c_j \sim \N(0, (j! \, l^{2j})^{-1})$. The covariance function of this class is now given as
\begin{equation}
\begin{split}
  K(x,y)
  &= \sum_{i} \frac{1}{i! \, \ell^{2i}} \, x^i \, \exp\left(-\frac{x^2}{2\ell^2} \right)
  \, y^i \, \exp \left(-\frac{y^2}{2\ell^2} \right) \\
  &= \exp\left( \frac{x \, y}{\ell^2} \right) \, \exp\left(-\frac{x^2}{2\ell^2} \right)
  \, \exp \left(-\frac{y^2}{2\ell^2} \right) \\
  &= \exp\left(-\frac{(x - y)^2}{2\ell^2} \right),
\end{split}
\end{equation}
which is the squared exponential covariance function. 

Based on the above, Minka \cite{Minka:2000} argued (although did not formally prove) that GPQs with the squared exponential covariance functions should converge to the classical quadratures. This argument is indeed backed up by our analytical example in Example~\ref{ex:gpt_se} where this covergence indeed happens.

\subsection{Random and quasi-random point sets}
Recall that one way to approximate the expectation of $\vg(\vxi)$ over a Gaussian distribution $\N(\mathbf{0},\MI)$ is to use Monte Carlo integration. In that method we simply draw $N$ samples from the Gaussian distribution $\xi_i \sim \N(\mathbf{0},\MI)$ and use them as sigma-points. The classical Monte Carlo approximation to the integral would now correspond to setting $W_i = 1/N$. Alternatively, we could use these random points as sigma-points and evaluate their weights by Equation~\eqref{eq:gp_wi}. This leads to an approximation, which is sometimes called the Bayesian Monte Carlo approximation \cite{OHagan:1987,Ghahramani:2002}.

Instead of sampling from the normal distribution, we can also use quasi-random points sets such as the Hammersley point sets \cite{Hammersley:1960,Hammersley:1964}. These are points sets which are designed to give a smaller error in average than random points. The classical method would correspond to setting all weights to $W_i = 1/N$, but again, we can also use Equation~\eqref{eq:gp_wi} to evaluate the weights for the GP quadrature. This corresponds to a "Bayesian quasi Monte Carlo" approximation to the integral. Some examples of Hammersley point sets are shown in Figure~\ref{fig:h_sigmas}.

\begin{figure}[!t]
\centering
\subfloat[3 points]{\includegraphics[width=0.2\textwidth]{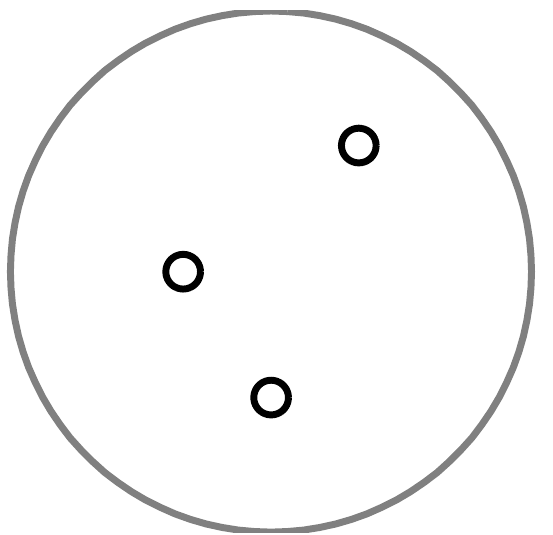}%
\label{fig:h3_sigmas}}
\hfil
\subfloat[7 points]{\includegraphics[width=0.2\textwidth]{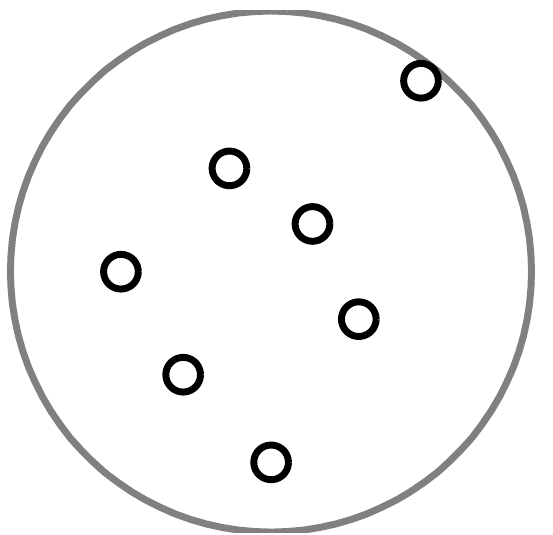}%
\label{fig:h7_sigmas}} \\
\subfloat[10 points]{\includegraphics[width=0.2\textwidth]{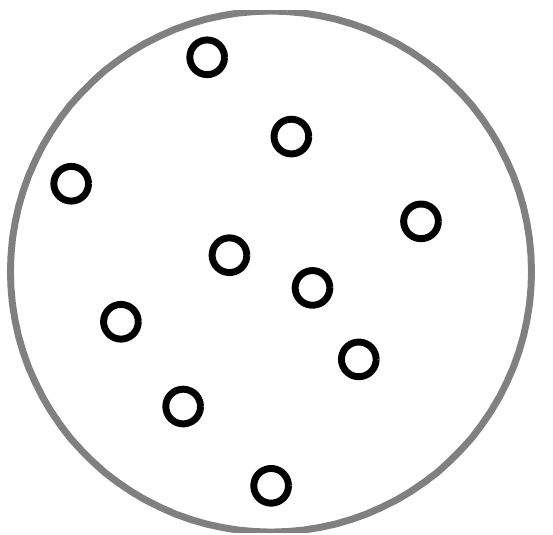}%
\label{fig:h10_sigmas}}
\hfil
\subfloat[20 points]{\includegraphics[width=0.2\textwidth]{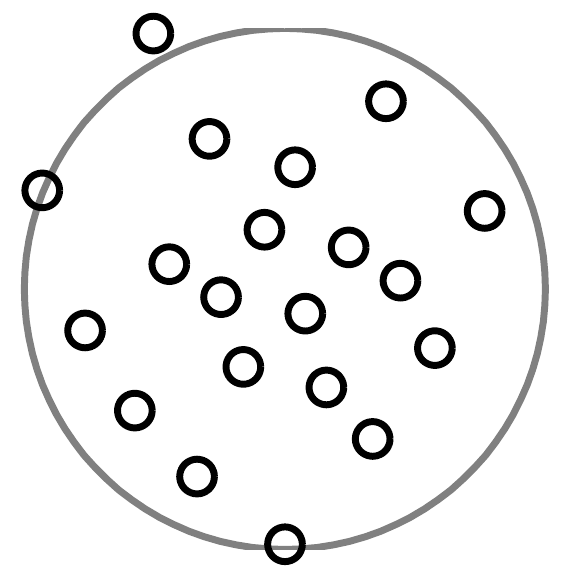}%
\label{ig:h20_sigmas}}
\caption{Hammersley point sets.}
\label{fig:h_sigmas}
\end{figure}

\section{Numerical Results}
\subsection{Covariance functions and regression implied by unscented transform}

\begin{figure}[!t]
\centering
\subfloat[UT-3]{\includegraphics[width=0.24\textwidth]{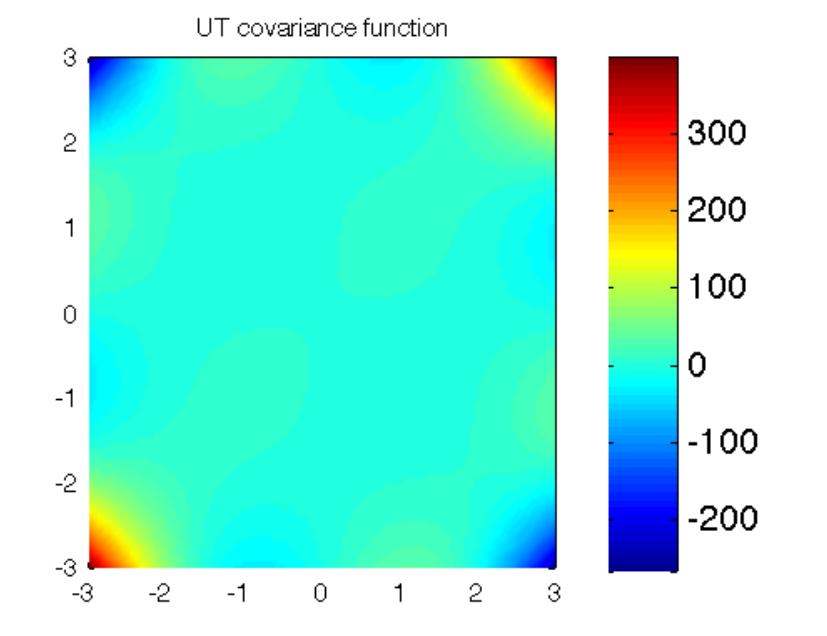}%
\label{fig:covmatrix_ut3}}
\hfil
\subfloat[UT-5]{\includegraphics[width=0.24\textwidth]{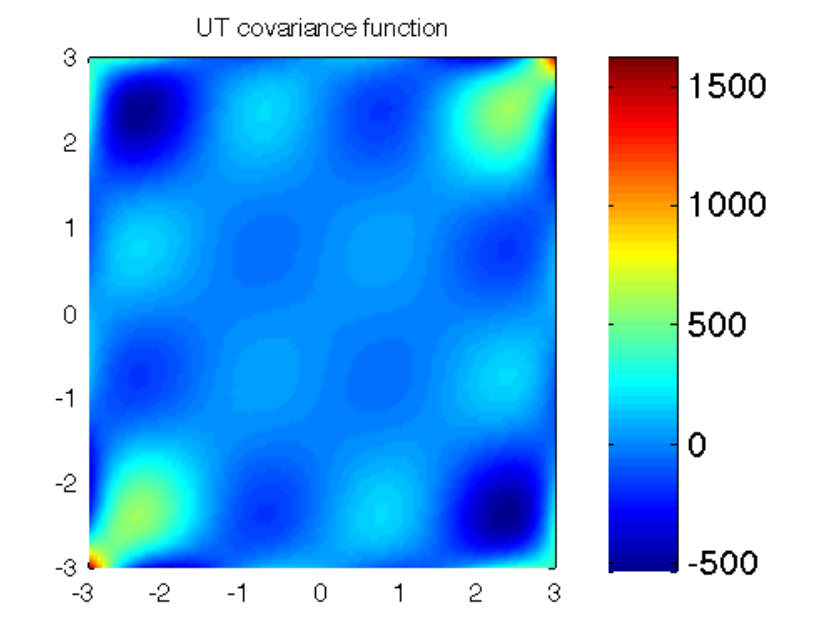}%
\label{fig:covmatrix_ut5}} \\
\subfloat[UT-7]{\includegraphics[width=0.24\textwidth]{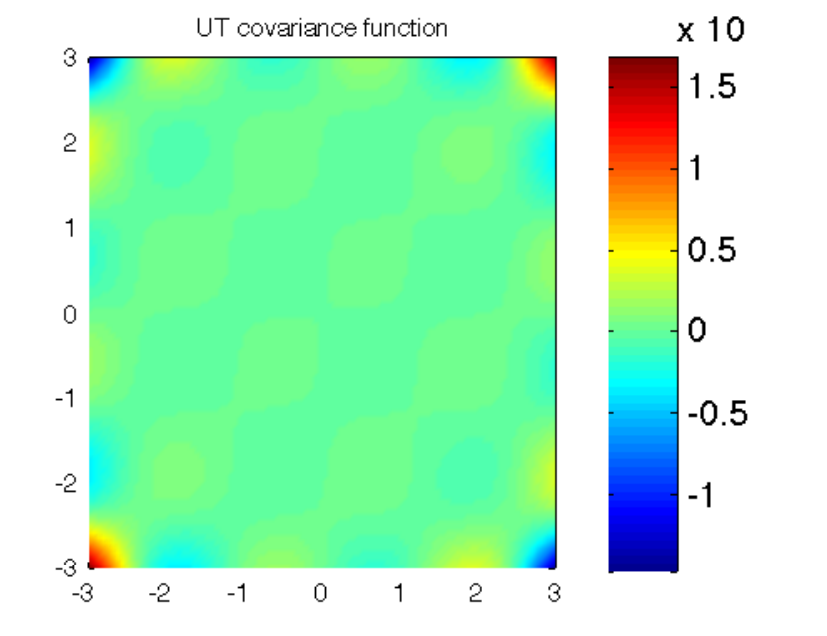}%
\label{fig:covmatrix_ut7}}
\hfil
\subfloat[SE]{\includegraphics[width=0.24\textwidth]{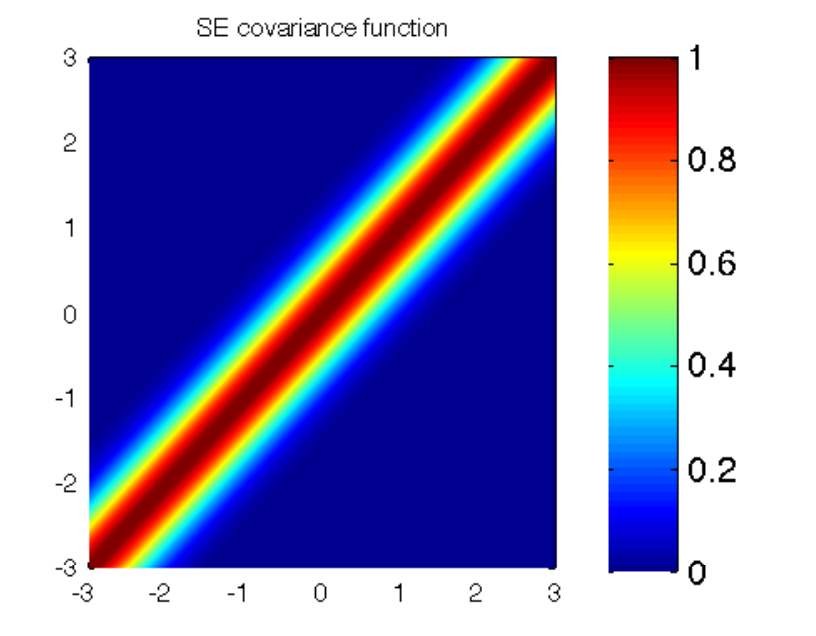}%
\label{ig:covmatrix_se}}
\caption{Covariance functions corresponding to different orders of unscented transforms (UT) and the squared exponential (SE) covariance function ($s = 1$, $\ell = 1/2$) for a single-input scalar-valued Gaussian process.}
\label{fig:covmatrix}
\end{figure}

\begin{figure}[!t]
\centering
\subfloat[UT-3]{\includegraphics[width=0.24\textwidth]{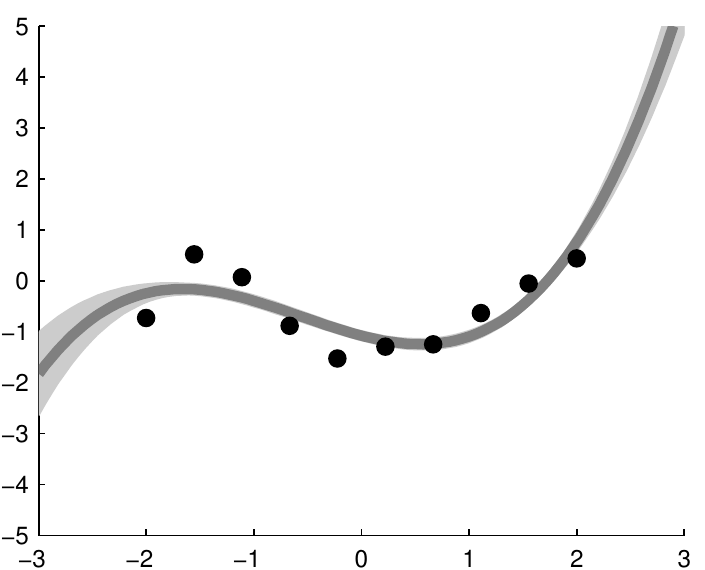}%
\label{fig:regr_ut3}}
\hfil
\subfloat[UT-5]{\includegraphics[width=0.24\textwidth]{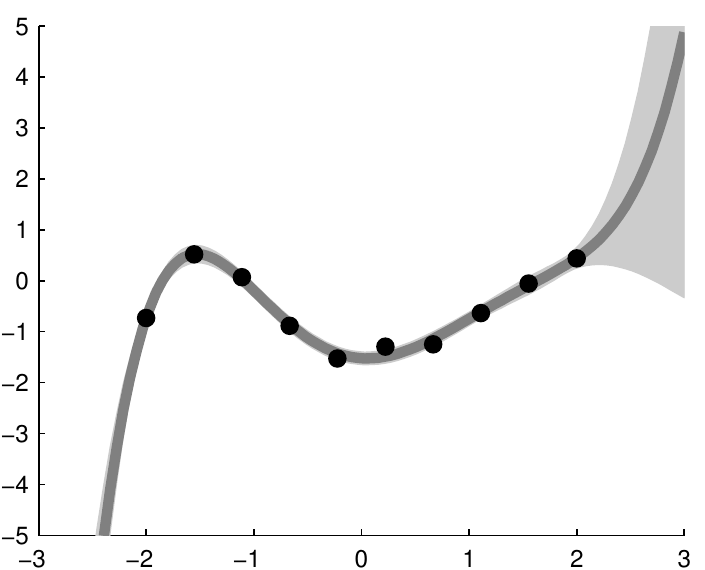}%
\label{fig:regr_ut5}} \\
\subfloat[UT-7]{\includegraphics[width=0.24\textwidth]{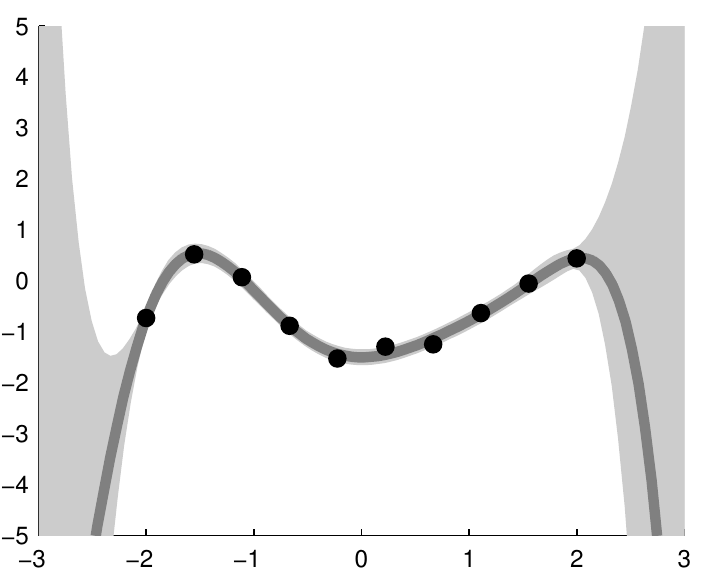}%
\label{fig:regr_ut7}}
\hfil
\subfloat[SE]{\includegraphics[width=0.24\textwidth]{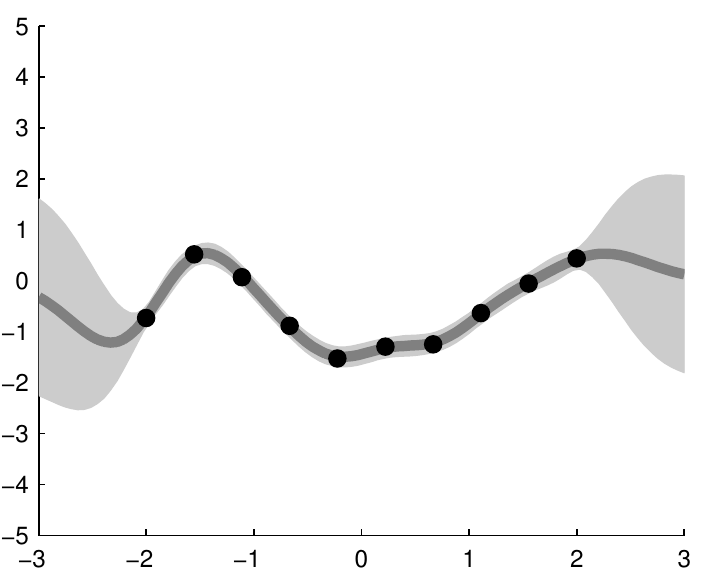}%
\label{ig:regr_se}}
\caption{Regression with covariance functions for UT and SE.}
\label{fig:regr}
\end{figure}

The unscented transform covariance functions of orders 3--7 (see Theorems \ref{the:ut3_cov} and \ref{the:utn_cov}) and the exponentiated quadratic (i.e., the squared exponential, SE) covariance function (Eq. \eqref{eq:secov}) are illustrated in Fig.~\ref{fig:covmatrix}. The polynomial nature of the unscented transform (UT) covariance function can be clearly seen in the figures -- the UT covariance function as such does not have such a simple local-correlation-interpretation as the SE covariance function has as the UT covariance functions simply blow up polynomially when moving away from the diagonal. 

The corresponding Gaussian process regression results on random data are illustrated in Fig.~\ref{fig:regr}. The polynomial nature of the unscented transform can be clearly seen in the figures. The Gaussian process prediction with the unscented transform covariance function has a clear polynomial shape as expected. Clearly the polynomial fit has less flexibility to explain the data than the exponentiated quadratic fit although the flexibility certainly grows with the polynomial (and thus UT) order.

\subsection{Illustrative high-dimensional example}
\begin{figure*}[!t]
\centering
\subfloat[Cubature for \eqref{eq:ara1}]{\includegraphics[width=0.33\textwidth]{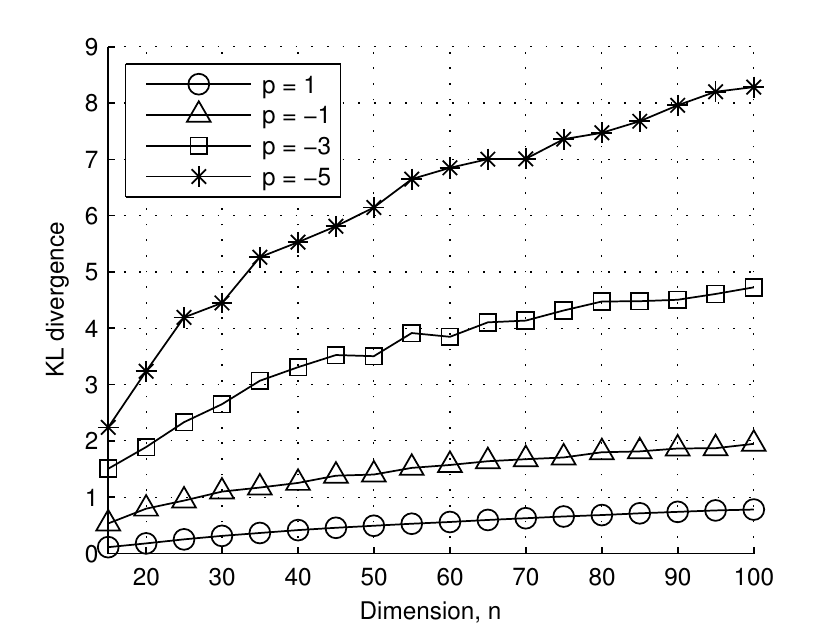}%
\label{fig:f1_ct}}
\hfil
\subfloat[GPQ-Cubature for \eqref{eq:ara1}]{\includegraphics[width=0.33\textwidth]{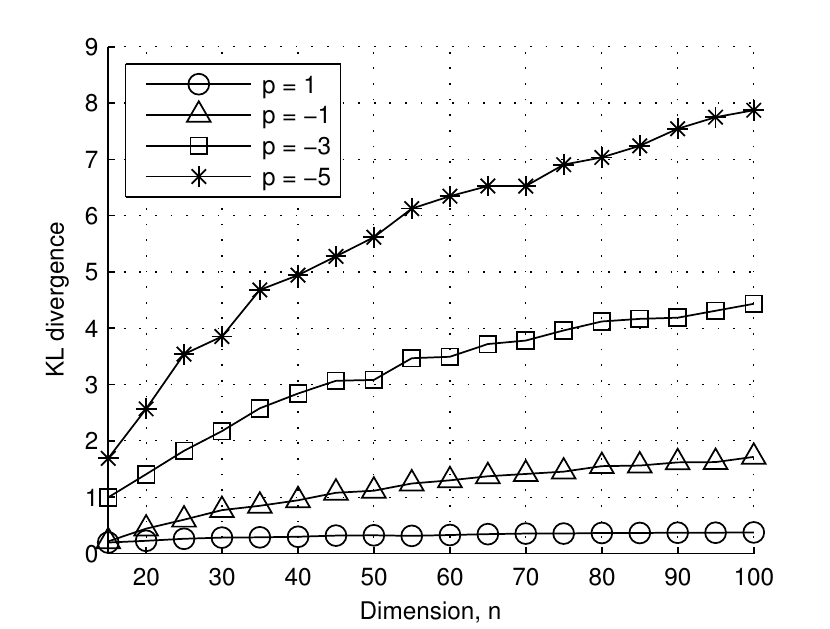}%
\label{fig:f1_gpc}}
\hfil
\subfloat[GPQ-Hammersley for \eqref{eq:ara1}]{\includegraphics[width=0.33\textwidth]{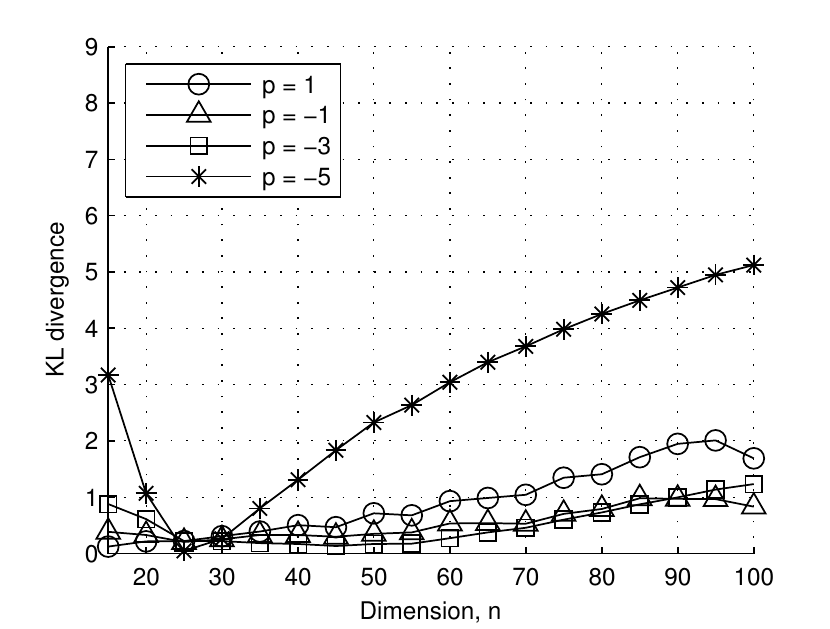}%
\label{fig:f1_gpc}} \\
\subfloat[Cubature for \eqref{eq:ara2}]{\includegraphics[width=0.33\textwidth]{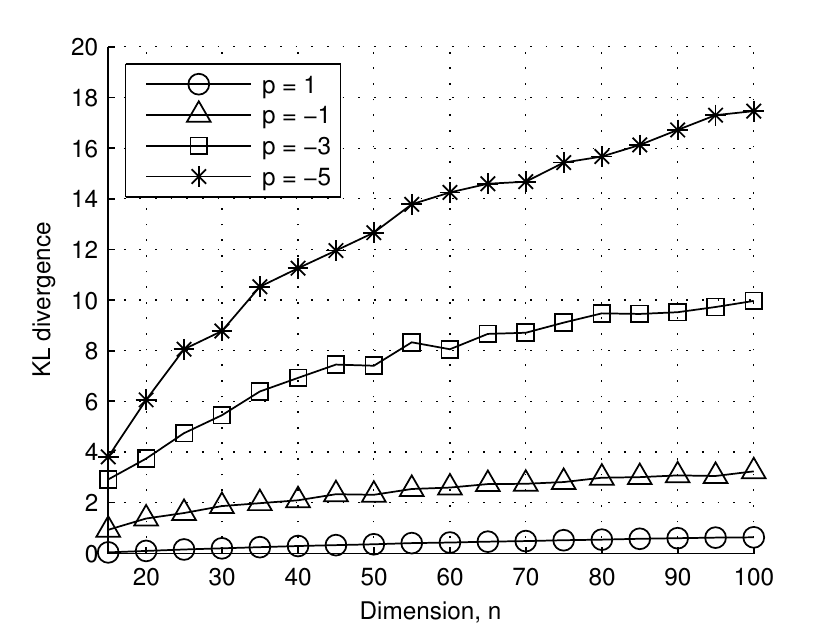}%
\label{fig:f2_ct}}
\hfil
\subfloat[GPQ-Cubature for \eqref{eq:ara2}]{\includegraphics[width=0.33\textwidth]{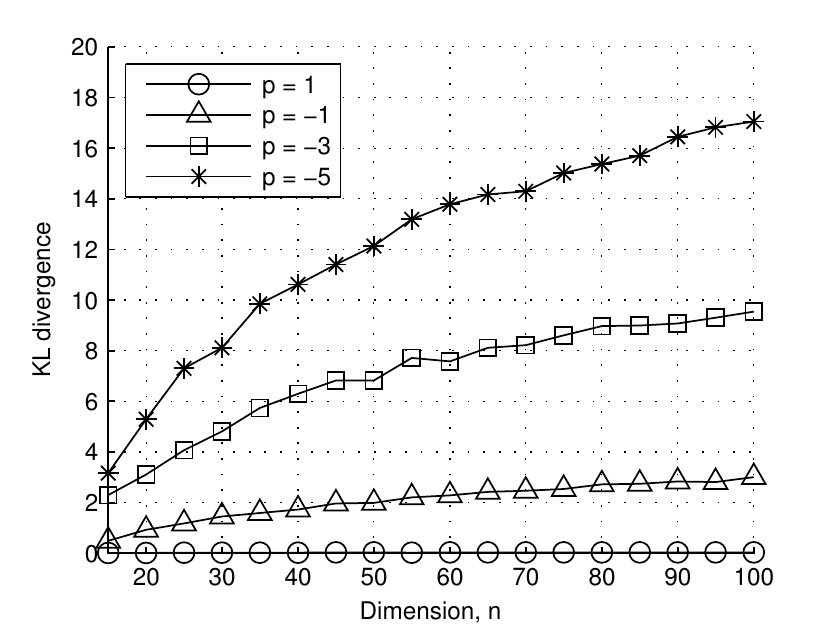}%
\label{fig:f2_gpc}}
\hfil
\subfloat[GPQ-Hammersley for \eqref{eq:ara2}]{\includegraphics[width=0.33\textwidth]{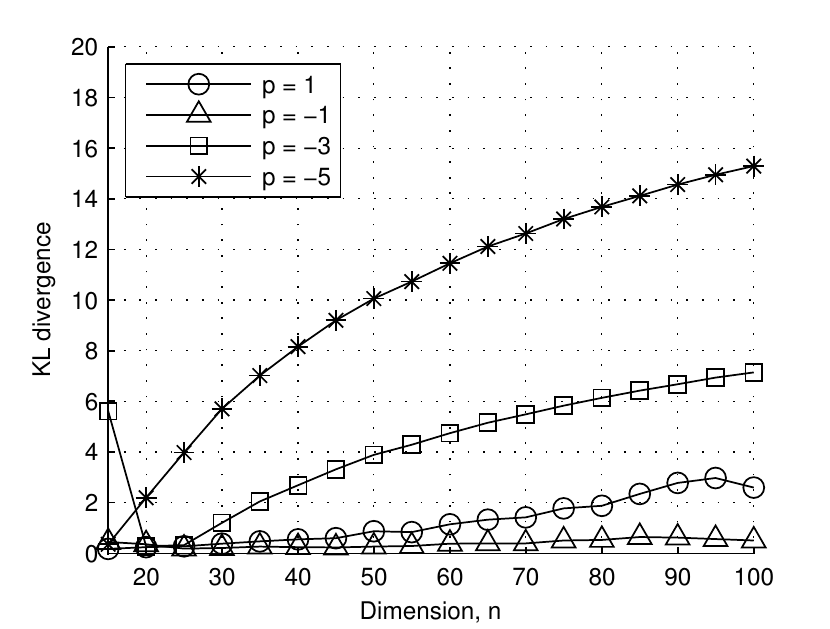}%
\label{fig:f2_gph}}
\caption{Comparison of different methods in computing the moment integrals used in \cite{Arasaratnam+Haykin:2009} for illustrating the performance of the cubature rule. It can be seen that the GPQ methods outperform the cubature rule in most of the cases.}
\label{fig:ct}
\end{figure*}

We use the same test case as in Section~VIII.A. of \cite{Arasaratnam+Haykin:2009}, that is, the computation of the first two moments of the function $y(\vx) = ( \sqrt{1 + \vx^\T \vx})^p$ for $p = 1,-2,-3,-5$. We thus aim to approximate the following integrals:
\begin{eqnarray}
  \mathrm{E}[y(\vx)] &=& \int \left( \sqrt{1 + \vx^\T \vx} \right)^p \, \N(\vx \mid \vm,\MP) \, \diff\vx,
\label{eq:ara1} \\
  \mathrm{E}[y^2(\vx)] &=& \int \left( 1 + \vx^\T \vx \right)^p \, \N(\vx \mid \vm,\MP) \, \diff\vx.
\label{eq:ara2}
\end{eqnarray}
Figure~\ref{fig:ct} shows the result of using the following methods as function of the state-dimensionality:
\begin{itemize}
\item {\em Cubature:} The 3rd order spherical cubature sigma-points ($2n$ points) with the standard integration weights.
\item {\em GPQ-Cubature:} The Gaussian process quadrature with SE covariance function and the 3rd spherical cubature sigma-points above.
\item {\em GPQ-Hammersley:} The Gaussian process quadrature with SE covariance and $2n$ Hammersley points.
\end{itemize}
The 3rd spherical cubature points refer to the integration rule proposed in \cite{McNamee+Stenger:1967}, which was also used in the cubature Kalman filter (CKF) in \cite{Arasaratnam+Haykin:2009}. In the rule, the sigma-points of placed to the intersections coordinate axes with the origin-centered $n$-dimensional hypersphere of radius $\sqrt{n}$. 

The results in Figure~\ref{fig:ct} show that the GPQ quite consistently gives a bit lower KL-divergence and hence better result than the plain cubature when the cubature points are used. When Hammersley point sets are used, the results vary a bit more: with small state dimensions the results are slightly worse than with the cubature points. When $p \ne 1$, the Hammersley results are much better in high dimensions whereas with $p = 1$ the results are worse than with the cubature point sets.

\subsection{Univariate non-linear growth model}
In this section we compare the performance of the different methods in the following univariate non-linear growth model (UNGM) which is often used for benchmarking non-linear estimation methods:
\begin{equation}
\begin{split}
  x_k &= \frac{1}{2} \, x_{k-1} + 25 \, \frac{x_{k-1}}{1 + x_{k-1}^2}
  + 8 \, \cos(1.2 \, k) + q_{k-1}, \\
  y_k &= \frac{1}{20} \, x_k^2 + r_k,
\end{split}
\end{equation}
where $x_0 \sim \N(0,5)$, $q_{k-1} \sim \N(0,10)$, and $r_k \sim \N(0,1)$.

We generated 100 independent datesets with 500 time step each and applied the following methods to it: extended, unscented ($\kappa = 2$), and cubature filters and smoothers (EKF/UKF/CKF/ERTS/URTS/CRTS); Gauss--Hermite filters and smoothers with 3, 7, and 10 points (GHKF/GHRTS); Gaussian process quadrature filter and smoother with unscented transform points (GPKFU/GPRTSU) and cubature points (GPKFC/GPRTSC); with Hammersley point sets of sizes 3, 7, and 10 (GPKFH/GPRTSH); and with minimum variance points sets of sizes 3, 7, and 10 (GPKFO/GPRTSO). The covariance function was the exponentiated quadratic with $s = 1$ and $\ell = 3$ and the noise variance was set to $10^{-8}$. The RMSE results together with single standard derivation bars are shown in Figures~\ref{fig:ungm_rmse_filters} and \ref{fig:ungm_rmse_smoothers}. As can be seen in the figures, with 5 and 10 points the Gaussian process quadrature based filters and smoothers have significantly lower errors than almost all the other methods --  only Gauss--Hermite with 10 points and the cubature RTS smoother come close. 

\begin{figure}[!htb]
\centering
\includegraphics[width=0.8\columnwidth]{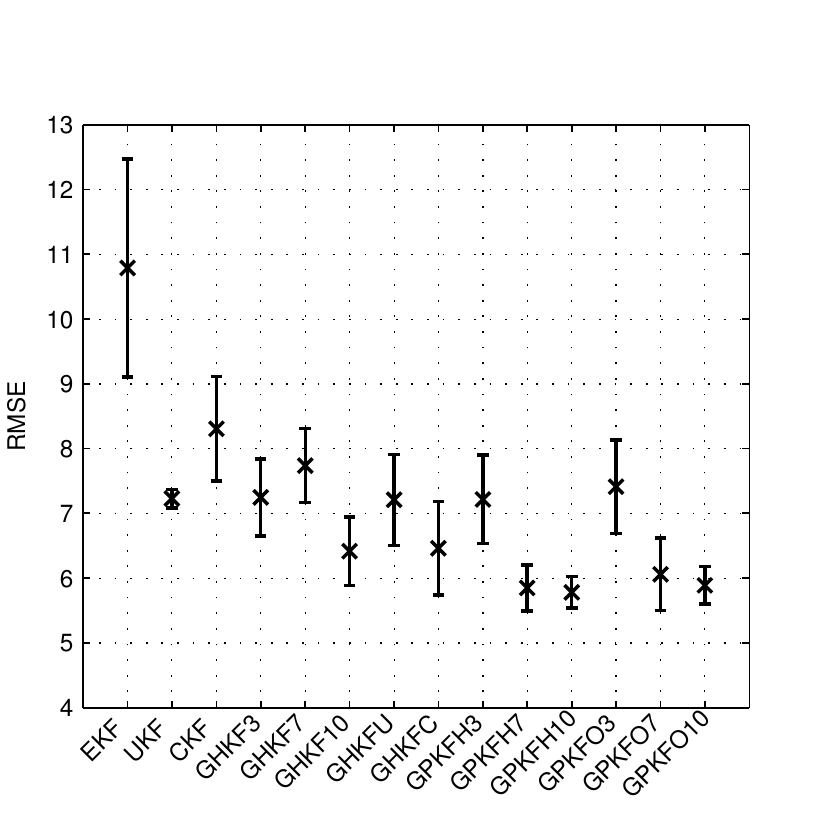}
\caption{RMSE results of filters in the UNGM experiment.}
\label{fig:ungm_rmse_filters}
\end{figure}

\begin{figure}[!htb]
\centering
\includegraphics[width=0.8\columnwidth]{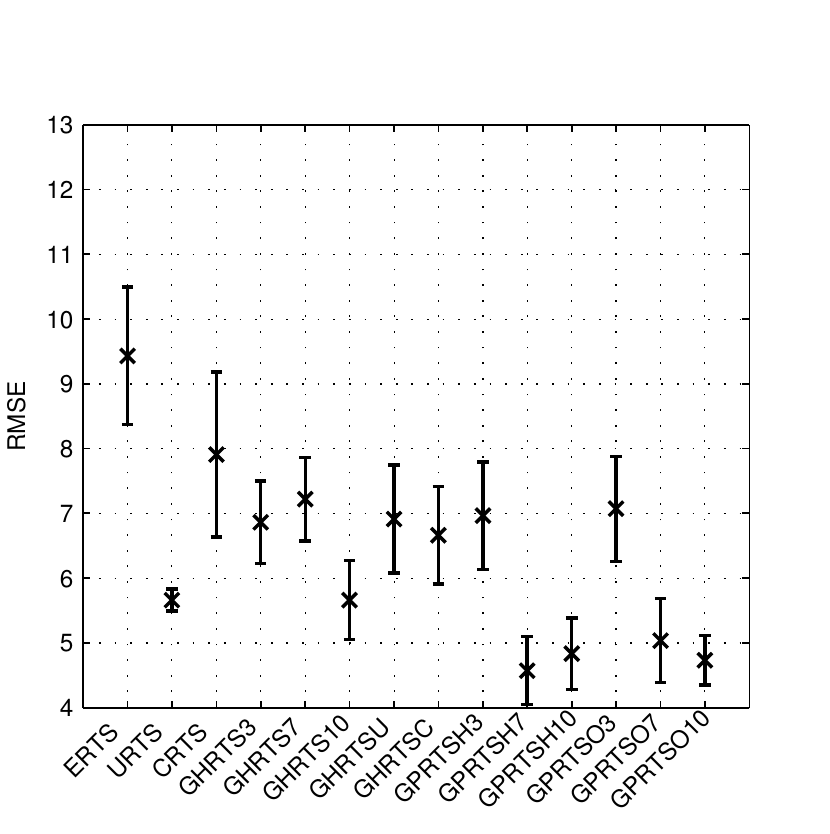}
\caption{RMSE results of smoothers in the UNGM experiment.}
\label{fig:ungm_rmse_smoothers}
\end{figure}

\subsection{Bearings only target tracking}
In this section we evaluate the methods in the bearings only target tracking problem with a coordinated-turn dynamic model, which was also considered in Section III.A of the article \cite{Sarkka+Hartikainen:2010a}. The non-linear dynamic model is
\begin{equation}
\vx_k = \begin{pmatrix} 1 & \frac{\sin(\omega_k \, \Delta t)}{\omega} & 0 & -
  \left( \frac{1-\cos(\omega_k \, \Delta t)}{\omega} \right)& 0 \\
0 & \cos(\omega_k \, \Delta t) & 0 & -\sin(\omega_k \, \Delta t) & 0 \\
0 & \frac{1-\cos(\omega_k \, \Delta t)}{\omega_k \, }& 1 & \frac{\sin(\omega
  \Delta t)}{\omega} & 0 \\
0 & \sin(\omega_k \, \Delta t) & 0 & \cos(\omega_k \, \Delta t) & 0 \\
0 & 0 & 0 & 0 & 1 \end{pmatrix} \vx_{k-1} + \vq_{k-1}, \label{eq:ct_model}
\end{equation}
where the state of the target is $\vx =
(x_1,\dot{x}_1,x_2,\dot{x}_2,\omega)$, and $x_1, x_2$ are the
coordinates and $\dot{x}_1, \dot{x}_2$ are the velocities in two
dimensional space.  The time step size is set to $\Delta t = 1 \;\text{s}$
and the covariance of the process noise $q_k \sim N(0,Q)$ is
\begin{equation}
Q = \begin{pmatrix}
q_1\frac{\Delta t3}{3} & q_1\frac{\Delta t2}{2} & 0 & 0 & 0 \\
q_1\frac{\Delta t2}{2} & q_1\Delta t & 0 & 0 & 0\\
0 & 0 & q_1\frac{\Delta t3}{3} & q_1\frac{\Delta t2}{2} & 0  \\
0 & 0 & q_1\frac{\Delta t2}{2} & q_1\Delta t & 0 \\
0 & 0 & 0 & 0 & q_2 \Delta t\\
\end{pmatrix},
\end{equation}
where we used $q_1 = 0.1 \text{m}2 \text{s}^{-3}$ and $q_2=1.75 \times 10^{-4} \text{s}^{-3}$.

In the simulation setup we have four sensors measuring the angles $\theta$
between the target and the sensors. The non-linear measurement model
for sensor $i$ can be written as
\begin{equation}
\theta^i = \arctan \left( \frac{x_2 - s^i_2}{x_1 - s^i_1}\right) +
r^i, \label{eq:bot_model}
\end{equation}
where $(s^i_1,s^i_2)$ is the position of the sensor $i$ in two
dimensions, and $r^i \sim N(0,\sigma^2_{\theta})$ is the
measurement noise. The used parameters were the same as in the article \cite{Sarkka+Hartikainen:2010a}.

The RMSE results for the position errors are shown in Figures \ref{fig:botm_rmse_filters} and \ref{fig:botm_rmse_smoothers}. Clearly all of the sigma-point methods outperform the Taylor series based methods (EKF/EKS). However, the performances of all the sigma-point methods are very similar: also the Gaussian process quadrature methods give very similar results to the other sigma-point methods. There is a small dip in the errors at the Gauss--Hermite based methods as well as in the highest order Hammersley GPQ method, but practically the performance of all the sigma-point methods is the same.

\begin{figure}[!htb]
\centering
\includegraphics[width=0.8\columnwidth]{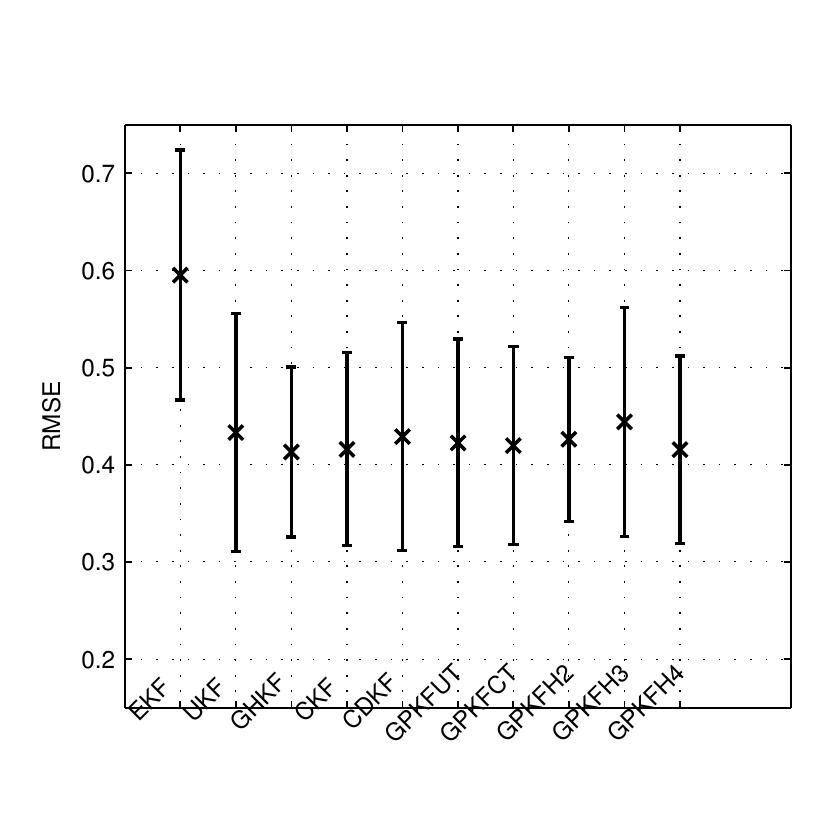}
\caption{Position RMSE results of filters in the bearings only experiment.}
\label{fig:botm_rmse_filters}
\end{figure}

\begin{figure}[!htb]
\centering
\includegraphics[width=0.8\columnwidth]{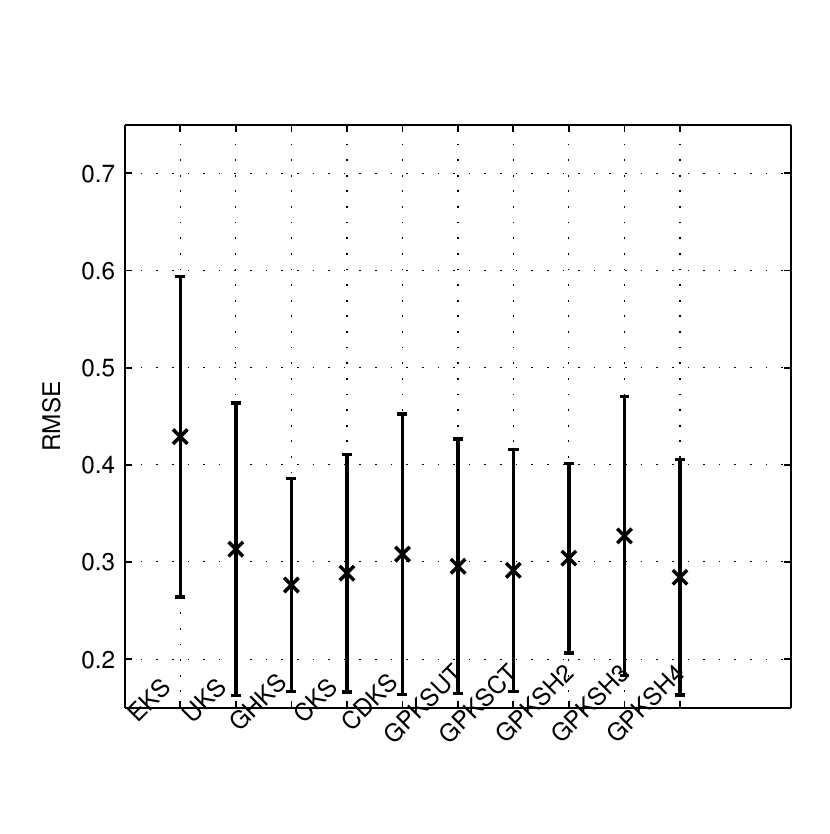}
\caption{Position RMSE results of smoothers in the bearings only experiment.}
\label{fig:botm_rmse_smoothers}
\end{figure}

\section{Conclusion}
In this article we have proposed new Gaussian process quadrature based non-linear Kalman filtering and smoothing methods and analyzed their relationship with other sigma-point filters and smoothers. We have also discussed the selection of the evaluation points for the quadratures with respect to different criteria: exactness for multivariate polynomials up to a given order, minimum average error, and quasi-random point sets. We have shown that with suitable selections of (polynomial) covariance functions for the Gaussian processes the filters and smoothers reduce to unscented Kalman filters of different orders as well as to Gauss--Hermite Kalman filters and smoothers. By numerical experiments we have also shown that the Gaussian process quadrature rules as well as the corresponding filters and smoothers often outperform previously  proposed (polynomial) integration rules and sigma-point filters and smoothers.

\appendices
\section{Fourier--Hermite series} \label{sec:fh}
Fourier--Hermite series (see, e.g., \cite{Malliavin:1997}) are orthogonal polynomial series in a Hilbert space, where the inner product is defined via an expectation of a product over a Gaussian distributions. These series are also inherently related to non-linear Gaussian filtering as they can be seen as generalizations of statistical linearization and they also have a deep connection with unscented transforms, Gaussian quadrature integration, and Gaussian process regression \cite{Sarmavuori+Sarkka:2012,Sandblom+Svensson:2012,Sarkka:2013}.

We can define an inner product of multivariate scalar functions $f(\vx)$ and $g(\vx)$ as follows:
\begin{equation}
  \langle f,g \rangle = \int f(\vx) \, g(\vx) \, \N(\vx \mid \mathbf{0},\MI) \, \diff \vx.
\end{equation}
If we now define a norm via $||f||_{\mathcal{H}}^2 = \langle f,f
\rangle$, and the corresponding distance function $d(f,g) = ||f -
g||_{\mathcal{H}}$, then the functions $||f||_{\mathcal{H}} < \infty$ form a Hilbert space $\mathcal{H}$. It now turns out that the
multivariate Hermite polynomials form a complete orthogonal basis of
the resulting Hilbert space \cite{Malliavin:1997}.

A multivariate Hermite polynomial with multi-index $\mathcal{I} = \{
i_1,\ldots,i_n \}$ can be defined as
\begin{equation}
    H_{\mathcal{I}}(\vx) 
    = H_{i_1}(x_1) \times \cdots \times H_{i_n}(x_n)
\end{equation}
which is a product of univariate Hermite polynomials
\begin{equation}
  H_p(x) = (-1)^p \, \exp(x^2/2) \, \frac{\diff^p}{\diff x^p} \exp(-x^2/2).
\end{equation}
The orthogonality property can now be expressed as
\begin{equation}
  \langle H_{\mathcal{I}}, H_{\mathcal{J}} \rangle
  = \begin{cases}
     \mathcal{I}!, & \text{if } \mathcal{I} = \mathcal{J} \\
     0, & \text{otherwise},
  \end{cases}
\end{equation}
where we have denoted $\mathcal{I}! = i_1! \cdots i_n!$ and
$\mathcal{I} = \mathcal{J}$ means that each of the elements in the
multi-indices $\mathcal{I} = \{ i_1,\ldots,i_n \}$ and $\mathcal{J} =
\{ j_1,\ldots,j_n \}$ are equal. We will also denote the sum of
indices as $|\mathcal{I}| = i_1 + \cdots + i_n$.

A function $g(\vx)$ with $\langle g, g \rangle < \infty$ can be
expanded into Fourier--Hermite series \cite{Malliavin:1997}
\begin{equation}
  g(\vx) = \sum_{p=0}^{\infty} \sum_{|\mathcal{I}| = p} 
  \frac{1}{\mathcal{I}!} c_{\mathcal{I}} \, H_{\mathcal{I}}(\vx),
\label{eq:fh}
\end{equation}
where $H_{\mathcal{I}}(\vx)$ are multivariate Hermite polynomials and the series coefficients are given by the inner products $c_{\mathcal{I}} = \langle H_{\mathcal{I}}, g \rangle$.

Consider a Gaussian process $g_G(\vx)$ which has zero mean and a
covariance function $K(\vx,\vx')$. In the same way as deterministic functions, Gaussian processes can also be expanded into Fourier--Hermite series:
\begin{equation}
  g_G(\vx) = \sum_{p=0}^{\infty} \sum_{|\mathcal{I}| = p} 
  \frac{1}{\mathcal{I}!} \tilde{c}_{\mathcal{I}} \, H_{\mathcal{I}}(\vx),
\label{eq:fhgp}
\end{equation}
where the coefficients are given as $\tilde{c}_{\mathcal{I}} = \langle H_{\mathcal{I}}, g_G \rangle$. The coefficients
$\tilde{c}_{\mathcal{I}}$ are zero mean Gaussian random variables and their covariance is given as
\begin{equation}
\begin{split}
  \E\left[\tilde{c}_{\mathcal{I}} \, \tilde{c}_{\mathcal{J}}\right]
  &= \E\left[\langle H_{\mathcal{I}}, g_G \rangle \,
        \langle H_{\mathcal{J}}, g_G \rangle\right] \\
  &= 
  \iint H_{\mathcal{I}}(\vx) \, 
  K(\vx,\vx') \, H_{\mathcal{J}}(\vx') 
  \\ &\qquad \times
       \N(\vx \mid \mathbf{0},\MI) \,
       \N(\vx' \mid \mathbf{0},\MI) \, \diff \vx \, \diff \vx'.
\end{split}
\end{equation}
If we define constants $\lambda_{\mathcal{I},\mathcal{J}} =
\E\left[\tilde{c}_{\mathcal{I}} \, \tilde{c}_{\mathcal{J}}\right]$
then the covariance function $K(\vx,\vx')$ can be further written as
series
\begin{equation}
  K(\vx,\vx') = 
  \sum_{q=0}^{\infty} \sum_{|\mathcal{J}| = q} 
  \sum_{p=0}^{\infty} \sum_{|\mathcal{I}| = p} 
  \frac{1}{\mathcal{I}! \, \mathcal{J}!} \lambda_{\mathcal{I},\mathcal{J}}
    \, H_{\mathcal{I}}(\vx) \, H_{\mathcal{J}}(\vx').
\end{equation}

\ifCLASSOPTIONcaptionsoff
  \newpage
\fi

%
%
%
%
%
%
%

%

%
%
%
%
%
%

%
%
%
%
%
%
%
%
%
%
%
%

%
%
%
%
%
%
%
%
%
%
%
%
%
%
%
%
%
%
%
%
%
%
%
%
%
%
%
%
%
%
%
%
%

%

%
%
%

%
\end{document}